\pdfoutput=1
\documentclass[12pt]{article}
\usepackage{graphicx}
\usepackage[sort,numbers]{natbib}  %
\usepackage{xurl} %
\usepackage{comment}
\newcommand{\blind}{0}

\let\oldbibliography\thebibliography
\renewcommand{\thebibliography}[1]{\oldbibliography{#1}
\setlength{\itemsep}{0pt}} 

\usepackage{authblk}
\usepackage{amsmath}
\usepackage{amssymb}
\usepackage{amsthm}
\usepackage{tikz}
\usetikzlibrary{arrows}
\usepackage{float} %
\usepackage{subcaption}

\usepackage{booktabs}
\usepackage{algorithm}
\usepackage{algorithmic}
\usepackage{bbm}
\usepackage{dsfont}
\usepackage{adjustbox}
\usepackage[shortlabels]{enumitem}
\usepackage{hyperref}
\usepackage[hang,flushmargin]{footmisc}
\usepackage{titlesec}

\usepackage{caption}
\makeatletter
\g@addto@macro\normalsize{%
  \setlength\abovedisplayskip{3pt}
  \setlength\belowdisplayskip{3pt}
  \setlength\abovedisplayshortskip{3pt}
  \setlength\belowdisplayshortskip{3pt}
}
\makeatother

\newcommand{\cF}{\mathcal{F}}
\newcommand{\cE}{\mathcal{E}}
\newcommand{\cG}{\mathcal{G}}

\newcommand{\cI}{\mathcal{I}}
\newcommand{\cS}{\mathcal{S}}
\newcommand{\cB}{\mathcal{B}}
\newcommand{\cC}{\mathcal{C}}
\newcommand{\cN}{\mathcal{N}}

\newcommand{\bdelta}{\boldsymbol{\delta}}
\newcommand{\bTheta}{\boldsymbol{\Theta}}

\newcommand{\cO}{\mathcal{O}}

\newcommand{\bbE}{{\mathbb{E}}}

\newcommand{\Data}{{\mathcal{D}_n}}

\newtheorem{theorem}{Theorem}

\newtheorem{lemma}[theorem]{Lemma}

\newcommand{\bbR}{{\mathbb{R}}}
\newcommand{\bbN}{{\mathbb{N}}}

\begin{document}

\def\spacingset#1{\renewcommand{\baselinestretch}%
{#1}\small\normalsize} \spacingset{1}

\titlespacing{\section}{0pt}{0pt}{0pt}
\titlespacing{\subsection}{0pt}{0pt}{0pt}
\titlespacing{\subsubsection}{0pt}{0pt}{0pt}

\if0\blind
{
  \title{\bf Isotonic surrogate modeling for computer experiments with many input variables}
  \small
   \author[1]{Jaehoan Kim}
   \author[1]{Simon Mak\footnote{This work is supported by funding from NSF DMS 2316012 and DE-AC02-05CH11231.}}
    \affil[1]{Department of Statistical Science, Duke University}
  \maketitle
} \fi

\if1\blind
{
  \bigskip
  \bigskip
  \bigskip
  \begin{center}
    {\LARGE\bf Isotonic surrogate modeling for computer experiments with many input variables}
\end{center}

  \medskip
} \fi

\bigskip
\begin{abstract}
Virtual simulators are widely used for studying complex physical phenomena, from particle collisions to rocket propulsion. Such ``computer experiments'' can be highly time-intensive, and a Bayesian surrogate model can be used for efficient emulation with reliable uncertainty quantification. To train accurate surrogates with a limited sample size $n$, recent work has explored the incorporation of monotonicity (or isotonicity) information, which can often be elicited from physical systems. In practical applications with many input variables, however, existing Bayesian isotonic models can face statistical and computational limitations, which may result in worse performance compared to models that do not incorporate isotonicity. We propose a new transformed additive isotonic model (TAIM), which aims to tame this ``curse-of-dimensionality''. TAIM makes use of a flexible transformed additive isotonic modeling framework, which leverages a data-estimated link transformation and a monotone basis model with spike-and-slab priors on basis weights. Prediction-wise, TAIM achieves (up to log factors) a posterior contraction rate of $\mathcal{O}(n^{-1/3})$ when the true black-box function is in a transformed additive isotonic form with mild smoothness conditions. Such a rate does not depend on the input dimension $d$ for terms involving $n$, which softens the effect of dimensionality on posterior predictions. Computation-wise, TAIM allows for efficient posterior inference via a carefully designed Gibbs sampler, where each sampling iteration requires only linear work in $d$. We further present an extension of TAIM that can model potential deviations from transformed additivity. Numerical experiments and two applications show the effectiveness of TAIM for isotonic surrogate modeling with many input variables.

\end{abstract}

\noindent%
\vfill

\newpage
\spacingset{1.7} %

\section{Introduction}

Scientific experimentation is undergoing rapid developments. With breakthroughs in scientific computing, complex phenomena, e.g., particle collisions \cite{ehlers2024bayesian}, aircraft propulsion \cite{miller2024expected} and nuclear detectors \cite{kim2025efficient}, can now be reliably simulated in a virtual environment. These ``computer experiments'' \cite{gramacy2020surrogates} are quickly replacing physical experiments in broad scientific disciplines. To faithfully capture reality, however, each computer experiment run can be highly complex and time-intensive, requiring thousands of CPU hours to perform on supercomputing systems \cite{ji2023graphical}. This is exacerbated by the high-dimensional nature of modern scientific problems, where many input variables are investigated over a high-dimensional domain $\mathcal{X} \subseteq \mathbb{R}^d$. Given limited computing resources, the exploration of the simulated response surface $f(\cdot)$ over this high-dimensional domain $\mathcal{X}$ is a challenging task.

A promising solution is surrogate modeling \cite{gramacy2020surrogates}. The idea is to first run the simulator at a set of designed input points $\mathbf{x}_1, \cdots, \mathbf{x}_n \in \mathcal{X}$, then use the simulated data $\{f(\mathbf{x}_i)\}_{i=1}^n$ to train a probabilistic predictive model that ``emulates'' $f$ with uncertainty at untested input points. A popular surrogate model is the Gaussian process (GP; \cite{sacks1989design,mitchell1992bayesian}). GP surrogates provide a flexible model with reliable uncertainty quantification (UQ) and closed-form predictive equations that facilitate efficient downstream analysis, e.g., inverse problems \cite{everett2021multisystem} and optimization \cite{jones1998efficient}. Recent work has explored deep GP \cite{sauer2023active} and multi-physics GP \cite{li2023additive} surrogates. Despite this body of work, however, GP surrogates can suffer from a ``curse-of-dimensionality'' \cite{van2008rates}, in that their predictions can deteriorate quickly as the number of inputs $d$ increases. Given a time-intensive computer simulator and thus a limited sample size $n$, GP surrogates can perform poorly for high-dimensional problems.

A saving grace is that, in many physical applications, one can elicit known ``physics'' for enhancing the surrogate model. This field of ``physics-informed machine learning'' \cite{willard2020integrating} is a promising emerging area. For GPs, recent work has explored the incorporation of mechanistic equations \cite{Wheeler14,chen2022apik}, boundary conditions \cite{ma2025constrained,ding2019bdrygp}, and isotonicity. We focus on the latter in this work. Here, isotonicity refers to the monotone (increasing\footnote{In this paper, ``monotone'' refers to ``monotone increasing'' for brevity. One can easily model for monotone decreasing behavior by applying the transformation $x_l \leftarrow 1 - x_l$ to the appropriate inputs.}) nature of $f$ for a ``monotone set'' of inputs $\mathcal{I} \subseteq \{1, \cdots, d\}$, i.e., for any choice of inputs $\mathbf{x} = (\mathbf{x}_{\mathcal{I}},\mathbf{x}_{\mathcal{I}^c})$:
\begin{equation}
f(\mathbf{x}_{\mathcal{I}},\mathbf{x}_{\mathcal{I}^c}) \leq f(\mathbf{x}_{\mathcal{I}}',\mathbf{x}_{\mathcal{I}^c}) \quad \text{whenever} \quad \mathbf{x}_{\mathcal{I}} \preceq \mathbf{x}_{\mathcal{I}}',
\label{eq:monotone}
\end{equation}
where $\mathbf{x} \preceq \mathbf{x}'$ implies elementwise $\leq$. Such information can often be deduced from physical intuition; e.g., in simulating a dropping ball, its height from the ground $f(x)$ will be monotone decreasing as time $x$ increases. Similar isotonicity properties can be elicited in complex computer models (see, e.g., \cite{bayarri2009predicting}). There is a body of work in the computer experiments literature on incorporating isotonicity for GP surrogates: \cite{golchi2015monotone,wang2016estimating} make use of conditioning on the derivative process being non-negative at discrete points, whereas \cite{maatouk2017gaussian,maatouk2025efficient} employ a basis expansion approach with appropriately constrained basis weights. 

There is also a rich literature on the broad topic of isotonic regression. Frequentist approaches are explored, e.g., in \cite{han2019isotonic,deng2020isotonic,fokianos2020integrated}, although these may not be suitable for our surrogate modeling problems where Bayesian UQ is often needed \cite{everett2021multisystem}. Bayesian approaches largely fall into two classes. The first is \textit{projection-based} approaches, where one places a prior stochastic process on $f$, then projects its posterior onto the space of isotonic functions. \cite{lin2014bayesian} makes use of this projection approach with a GP, \cite{saarela2011method,chakraborty2021convergence} employ a piecewise constant prior, and \cite{chakraborty2021convergence,wang2023posterior} prove posterior contraction rates in the univariate ($d=1$) and multivariate settings, respectively. The second is \textit{basis expansion} approaches, where one makes use of monotone basis functions with non-negative basis weights. \cite{scott2015nonparametric} uses such an approach with smoothing and regression splines, \cite{maatouk2017gaussian} utilizes correlated basis weight priors derived from a GP, \cite{ray2020efficient,maatouk2025efficient} explore Markov chain Monte Carlo (MCMC) algorithms for efficient posterior sampling, and \cite{zhou2024mass} investigates their posterior contraction rates in the univariate setting.

Existing Bayesian isotonic models, however, may face key limitations for \textit{high-dimensional} surrogate modeling. Computation-wise, these models typically require some form of discretization over the domain $\mathcal{X}$: projection-based approaches require sampling posterior paths over a grid on $\mathcal{X}$, and basis expansion approaches employ basis functions parametrized by knots over a grid on $\mathcal{X}$. Thus, as dimension $d$ increases, posterior sampling for both approaches can become \textit{prohibitively costly}. This is limiting for surrogate models, which should serve as an \textit{efficient} emulator for the expensive computer simulator. Prediction-wise, existing Bayesian isotonic models may suffer from a \textit{curse-of-dimensionality}, in that their prediction quality (in terms of posterior contraction rates) deteriorates rapidly in dimension $d$ \cite{wang2023posterior}. This is again limiting for surrogate models, where only a small sample size can be afforded. There is some work on leveraging an additive structure for isotonic modeling \cite{bacchetti1989additive,mammen2007additive}, but such structure can be too restrictive in applications \cite{goldstein2009common,shao2008genetic}. In realistic applications with many input variables, such limitations can result in worse performance compared to models that do not incorporate isotonicity, as we shall see later.

We propose a new Bayesian isotonic surrogate model, called the transformed additive isotonic model (TAIM), which addresses these limitations. TAIM adopts the model $f(\mathbf{x}) = g\{\sum_{l=1}^d h_l(x_l)\}$, where $g$ is a \textit{data-estimated} monotone link function that transforms the latent additive form $\sum_{l=1}^d h_l(x_l)$. Here, monotonicity of $f$ over inputs $l \in \mathcal{I}$ is ensured via monotone functions on $\{h_l\}_{l \in \mathcal{I}}$ and $g$. The use of a nonlinear transformation $g$ on an additive form is inspired by the widely-used generalized additive models \cite{hastie2017generalized} and Kolmogorov-Arnold networks \cite{liu2025kan}, which leverage similar transforms for flexible non-additive modeling in high dimensions. In contrast with the generalized additive shape-constrained model in \cite{chen2016generalized}, the monotone link $g$ here is neither fixed nor specified from a noise model, but rather \textit{estimated} from data. TAIM models $g$ and $\{h_l\}_{l=1}^d$ via basis expansions, with monotone basis functions and spike-and-slab basis weight priors \cite{ishwaran2005spike} for $g$ and $\{h_l\}_{l \in \mathcal{I}}$. We prove that TAIM achieves a posterior contraction rate of $\mathcal{O}(n^{-1/3})$ (up to log factors) when the true function $f_0$ is in a transformed additive form with mild smoothness conditions. Such a rate ``tames'' the curse-of-dimensionality for existing Bayesian isotonic models, whose rates can deteriorate rapidly in dimension $d$ \cite{wang2023posterior}. We then develop an efficient Gibbs sampler, where posterior samples can be obtained in work linear in $d$, enabling efficient Bayesian inference in high dimensions. We further extend this framework to model for potential deviations from transformed additivity. Finally, we investigate the performance of TAIM compared to existing models in numerical experiments and in two surrogate modeling applications.

The paper is organized as follows. Section \ref{sec: background} provides background and motivation. Section \ref{sec: taim} presents our model and investigates its posterior contraction rates. Section \ref{sec: computation} presents an efficient Gibbs sampler, and outlines an extension for modeling potential deviations from transformed additivity. Sections \ref{sec: simulation} and \ref{sec: application} describe our numerical experiments and applications. Section \ref{sec: conclusion} concludes the paper.

\section{Background and Motivation}
\label{sec: background}

For this paper, we presume that data are collected from the model:
\begin{equation}
y_i = f(\mathbf{x}_i) + \epsilon_i, \quad \epsilon_i \overset{i.i.d.}{\sim} \mathcal{N}(0,\sigma^2_\epsilon), \quad i = 1, \cdots, n.
\label{eqn: data}
\end{equation}
Here, $f$ is the unknown black-box function, $\{\mathbf{x}_i\}_{i=1}^n \subseteq \mathcal{X}$ are the designed input points, $\{\epsilon_i\}_{i=1}^n$ are the noise terms, and the domain $\mathcal{X} = [0,1]^d$ is taken as the unit hypercube. For computer experiments, this corresponds to the stochastic setting where simulated outputs are corrupted by noise \cite{gramacy2020surrogates}. We adopt such a setting as motivated by our applications in Section \ref{sec: application} and to align with existing literature on isotonic models; the deterministic setting is recovered with $\sigma^2_\epsilon = 0$. As mentioned earlier, our focus is on \textit{Bayesian} isotonic models that provide UQ for surrogate modeling. We first review such existing models, then inspect their potential limitations in high dimensions via a motivating experiment. 

\subsection{Projection-Based Approaches}

One class of Bayesian isotonic models is projection-based approaches. Here, one adopts a prior stochastic process on $f$, then projects its posterior distribution onto the space of isotonic functions. \cite{lin2014bayesian} employs such an approach with a GP prior on $f$. One first obtains the GP posterior by conditioning on training data $\mathcal{D}_n = \{(\mathbf{x}_i,y_i)\}_{i=1}^n$, then performs the following $L^2$-projection on sample paths $\tilde{f}$ from this posterior:
\begin{equation}
        \arg\min_{g \in \mathcal M}
    \int_{\mathcal X} \{g(\mathbf{x})-\tilde f(\mathbf{x})\}^2 d\mathbf{x},
    \label{eq:projection_operator}
\end{equation}
where $\mathcal{M}$ is the space of isotonic functions on $\mathcal{X}$. The projected sample paths then serve as posterior samples on the monotone-constrained process. \cite{saarela2011method,chakraborty2021convergence} explore related projection approaches using a piecewise constant prior model on $f$. 

Computation-wise, one challenge is that the projection step \eqref{eq:projection_operator} needs to be performed for many posterior sample draws over $\mathcal{X} = [0,1]^d$, which can be very costly in high dimensions. Specifically, posterior paths need to be sampled over a grid on $\mathcal{X}$ (with discretization level $N$ over each input), then projected by iteratively performing univariate projection steps over each monotone direction, using the pooled adjacent violators algorithm \cite{barlow1972statistical}. One can show that sampling each projected path requires $\mathcal{O}(N^{3d})$ work, which makes posterior sampling prohibitively costly in high dimensions (i.e., large $d$) with fine discretization (i.e., large $N$). In practice, a coarse discretization (i.e., small $N$) is thus needed for computational tractability in high dimensions. This can, however, considerably degrade model performance, resulting in \textit{worse} predictions compared to models that do not incorporate isotonicity, as we shall see later. 

Prediction-wise, such approaches may also suffer from a curse-of-dimensionality. \cite{wang2023posterior} showed that, with a piecewise constant prior on $f$, this projection-based model yields a posterior contraction rate of $\mathcal{O}(n^{-1/(2+d)})$ when the true function $f_0$ is in the space of monotone functions on $\mathcal{X}$. This can be limiting for high-dimensional surrogate modeling: to achieve a prediction error of $\epsilon > 0$, one needs a sample size of $n = \mathcal{O}(\epsilon^{-(2+d)})$, which is prohibitive for large $d$ given the cost of each experiment run. An alternate approach is thus needed for flexible isotonic modeling that softens this curse-of-dimensionality. 

Another related approach is the monotone-constrained GPs in \cite{wang2016estimating, golchi2015monotone}, which condition on $f$ having non-negative derivatives at finite ``derivative points'' on $\mathcal{X}$. This can be seen as a ``projection'' of the posterior distribution onto these finite points. Such models have shown promise for low-dimensional surrogate modeling applications \cite{wang2016estimating, golchi2015monotone}. However, they are similarly computationally expensive to fit for large $d$, since the number of derivative points needed to ensure monotonicity over the full domain $\mathcal{X}$ grows exponentially in $d$.

\subsection{Basis Expansion Approaches}

Another class of Bayesian isotonic models is basis expansion approaches. In the univariate setting of $d=1$, such models typically take the form:
\begin{equation}
    f(x) = \beta_0 + \sum_{j=1}^{N} \beta_j \psi_{j}(x),
    \label{eq:basis}
\end{equation}
where $\{\psi_j(x)\}_{j=1}^N$ are the monotone basis functions, and $\{\beta_j\}_{j=1}^N$ are their basis weights. Here, $N$ controls the number of knots, where more knots permit greater model flexibility. Monotonicity of $f$ is ensured with non-negative priors on basis weights. The modeling form \eqref{eq:basis} naturally extends to the multivariate setting of $d>1$ (see \cite{maatouk2017gaussian}), using tensor-product basis functions ${\{\phi_{j_1}(x_1) \cdots \phi_{j_d}(x_d)\}_{j_1=1}^{N_1} \cdots }_{j_d=1}^{N_d}$ (where $N_l$ controls the number of knots in dimension $l$) and appropriate constraints on basis weights to ensure isotonicity. \cite{maatouk2017gaussian} makes use of integrated hat basis functions with correlated basis weight priors derived from a GP. \cite{ray2020efficient,maatouk2025efficient} investigate MCMC algorithms for posterior sampling of such models. 

These models, however, suffer from similar limitations on computation and prediction in high dimensions. Computation-wise, such models require $\mathcal{O}(M^3)$ work per MCMC iteration, where the number of basis functions $M$ grows as $M = \prod_{l=1}^d N_l$. In $d=10$ dimensions, using $N_l=10$ in each dimension results in $M=10^{10}$ basis functions, which makes posterior sampling intractable. Using a smaller $N_l$ allows for computational tractability, but may result in worse predictions than if isotonicity were not incorporated, as we see later. Prediction-wise, there is little work to our knowledge on posterior contraction rates of such basis isotonic models for $d > 1$. However, related work in the frequentist setting (e.g., \cite{han2019isotonic}) suggests that a similar curse-of-dimensionality may be present for prediction.

\subsection{Motivating Experiment}

To inspect these limitations in high dimensions, we consider an application on the surrogate modeling of an aircraft wing simulator. Figure \ref{fig:Metric comparison 11D} (left) visualizes this set-up: the goal is to analyze the deflection on the main wing spar, which is subject to loads at ten stations along the spar. The virtual simulator has $d=11$ input variables (dynamic pressure $x_1$, lift multipliers at each station $x_2, \cdots, x_{11}$), and outputs the maximum vertical deflection over the stations. Here, physical knowledge suggests that the output $f$ is monotone over all $d=11$ inputs, i.e., $\mathcal{I} = \{1, \cdots, 11\}$. Further details are provided later in Section \ref{sec: application}.

We consider three Bayesian surrogate model benchmarks. The first is a standard GP with an anisotropic squared-exponential kernel, with parameters estimated via maximum likelihood. The second is the projection-based isotonic GP (Proj-GP) in \cite{lin2014bayesian}. The third is the basis expansion isotonic GP (Basis-GP) in \cite{maatouk2017gaussian}. For fair comparison, the grid resolutions for the latter two are set such that MCMC sampling does not exceed a time constraint of 30 minutes, which is already time-intensive for surrogate modeling. We compare these benchmarks with our TAIM and TAAIM models (introduced next), subject to the same time constraint. All models are trained on the same $n=10d=110$-point Latin hypercube design (LHD; \cite{gramacy2020surrogates}). Performance is evaluated on test MSE (for point predictions) and CRPS (a scoring rule for probabilistic predictions; \cite{gneiting2007strictly}). Section~\ref{ssec: wing spar} gives further details.

Figure \ref{fig:Metric comparison 11D} (right) shows the test log-MSE and CRPS boxplots over 20 replications. For both metrics, the existing isotonic models perform worse (in median error) than the standard GP, which does not incorporate isotonicity. A likely reason is that, given the time constraint on MCMC sampling, both models required a coarse discretization for high-dimensional model fitting, which can considerably hurt predictive performance. Our proposed models address this via a transformed additive framework that facilitates flexible and efficient isotonic modeling in high dimensions. As such, our models can better integrate isotonicity information here, resulting in better predictions than the standard GP.

\begin{figure}
    \centering
    \includegraphics[width=0.9\linewidth]{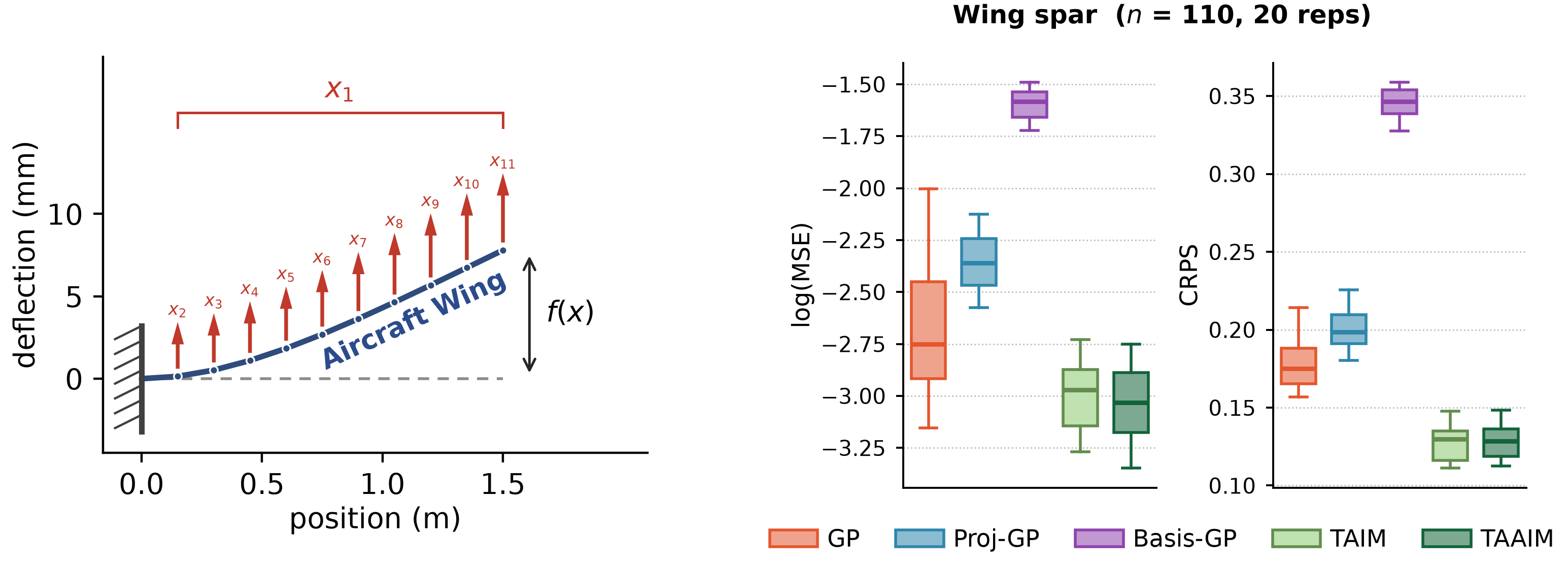}
    \caption{(Left) Visualizing the aircraft wing spar simulator set-up. (Right) Boxplots of the test log-MSE and CRPS for the aircraft wing application over 20 replications.}
    \label{fig:Metric comparison 11D}
\end{figure}

\section{Transformed Additive Isotonic Model}
\label{sec: taim}

\noindent We first present our TAIM modeling framework, then investigate its theoretical properties.

\subsection{Modeling Framework}\label{ssec: model}

We presume the data are generated from \eqref{eqn: data}, where $f$ is modeled as:
\begin{equation}\label{eqn: generalized additive model}
f(\mathbf{x}) = g\left\{\sum_{l=1}^d h_l(x_l)\right\}.
\end{equation}
Here, the additive functions $h_l$ are constrained to be monotone (increasing) if input $l$ is in the monotone set $\mathcal{I}$. The link transformation $g$ (to be estimated from data) is also constrained to be monotone. These constraints ensure that the desired isotonicity property \eqref{eq:monotone} is satisfied. One can easily account for monotone \textit{decreasing} behavior of an input $l$ by performing the transform $x_l \leftarrow 1 - x_l$ prior to fitting the above model.

This modeling form, i.e., the use of a nonlinear transformation $g$ on an additive structure, is motivated by a rich literature in statistics, dating back to the Box-Cox transformation \cite{box1964analysis}. Such a form is employed, e.g., in generalized additive models (GAMs; \cite{hastie2017generalized}), which are widely used for flexible and interpretable high-dimensional modeling. A key difference between a GAM and our model \eqref{eqn: generalized additive model} (besides the incorporation of isotonicity) is that the link $g$ in a GAM is typically presumed to be \textit{known}, e.g., from an underlying noise model. In our model, $g$ is an \textit{unknown} transformation to be estimated from data. Similar transformed additive models have been used in neural networks \cite{horowitz2007rate}, in Kolmogorov-Arnold networks \cite{liu2025kan} and for computer experiments \cite{lin2020transformation}, but to our knowledge, they have not been leveraged for high-dimensional isotonic modeling. Our model also contrasts with the generalized additive shape-constrained model in \cite{chen2016generalized}, for which $g$ is again presumed to be known.

For an input $l \in \mathcal{I}$, we adopt the following basis model for the monotone function $h_l$:
\begin{equation}
h_l(x)=\sum_{j=1}^{N_l}{\beta_j^{(l)}}{\psi_j^{(l)}}(x), \quad \quad l \in \cI, \quad \quad x \in [0,1].
\label{eqn: uniadd}
\end{equation}
Here, $\{{\psi_j^{(l)}}(x)\}_{j=1}^{N_l}$ are the basis functions for the $l$-th input, and $\{{\beta_j^{(l)}}\}_{j=1}^{N_l}$ their corresponding basis weights. We use piecewise linear monotone basis functions of the form:
\begin{equation}\label{eqn: psi definition}
    \psi_j^{(l)}(x) = \varrho(N_l x - j + 1), \quad \quad \varrho(z) = \min\{\max(z, 0), 1\}, \quad \quad j = 1, \cdots, N_l,
\end{equation}
where $\varrho(\cdot)$ ``clamps'' the function to be within $[0,1]$. The number of basis functions $N_l > 0$ controls the modeling resolution. Note that the basis function $\psi_j^{(l)}$ is centered at the knot $(2j - 1)/(2N_l)$, for $j = 1, \cdots, N_l$. Figure \ref{fig:monotone basis} (left) visualizes these basis functions for $N_l = 5$.

\begin{figure}[t]
    \centering
    \includegraphics[width=0.72\linewidth]{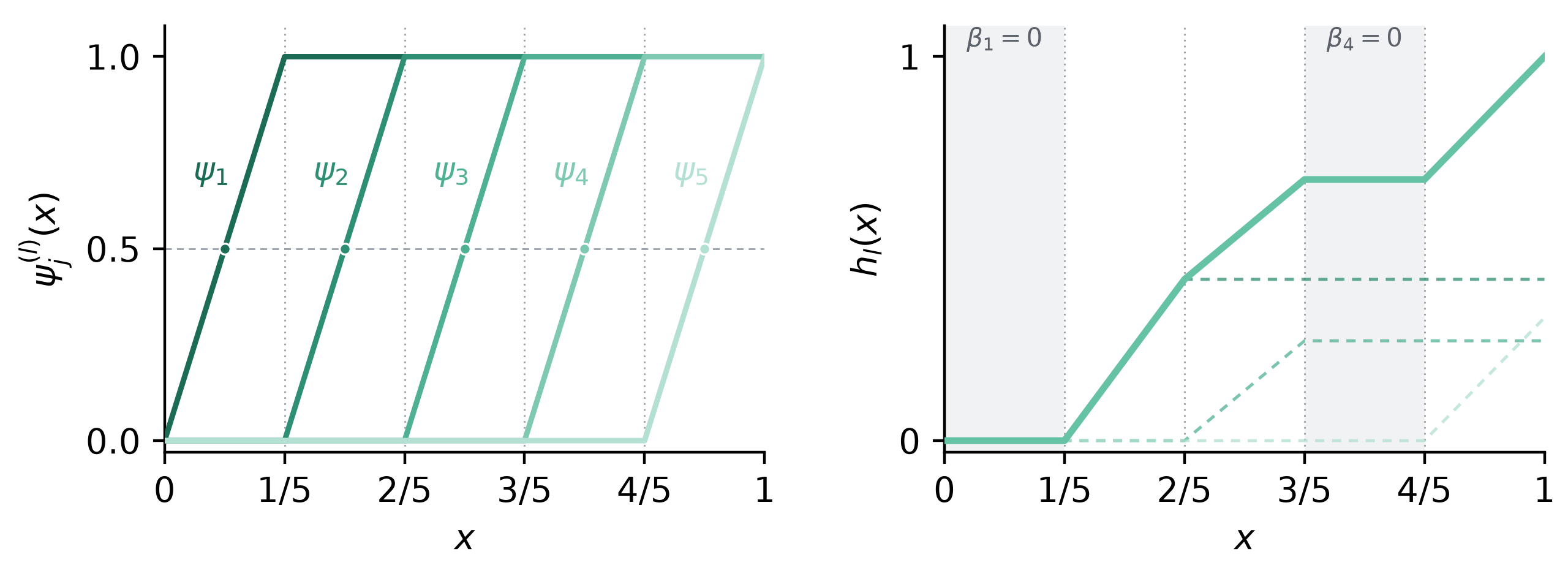}
    \caption{(Left) Visualizing the monotone basis functions \eqref{eqn: psi definition} with
    $N_l = 5$. (Right) Visualizing a monotone function modeled by these basis functions with non-negative basis weights. Dashed lines show the individual contributions of each basis term.}
    \label{fig:monotone basis}
\end{figure}

To enforce monotonicity on $h_l(x)$ for $l \in \cI$, its basis weights need to be non-negative (see Figure \ref{fig:monotone basis} right). We thus assign independent spike-and-slab priors \cite{ishwaran2005spike} of the form:
\begin{equation}
     \beta^{(l)}_j = I_{j}^{(l)} Z_j^{(l)}, \quad I_{j}^{(l)} \overset{\text{indep.}}{\sim} \text{Bern}\Big(\frac{\zeta_l}{N_l} \Big), \quad Z_j^{(l)} \overset{\text{indep.}}{\sim} \mathcal{N}^+(0, \tau_l^2), \quad j = 1, \cdots, N_l, \quad l \in \cI.
     \label{eqn: weightprior1}
\end{equation}
Here, each weight $\beta^{(l)}_j$ is the product of an indicator variable $I_{j}^{(l)}$ and a magnitude variable $Z_j^{(l)}$. The indicator $I_{j}^{(l)}$ follows a Bernoulli distribution with probability $\zeta_l/N_l$, where $\zeta_l$ is a prior hyperparameter. The magnitude $Z_j^{(l)}$ follows a half-normal distribution, i.e., a truncated normal distribution on $(0,\infty)$, with variance hyperparameter $\tau_l^2 > 0$. In other words, each weight is zero with probability $1-\zeta_l/N_l$, and half-normally-distributed (i.e., positive) otherwise. These spike-and-slab priors on basis weights ensure that $h_l$ has sufficient prior probability (specifically, ``small ball probability'' \cite{van2008rates}) over the space of monotone functions, which we need for proving the contraction rates later in Section \ref{sec: theory}.

For an input $l \notin \cI$, $h_l(x)$ need not be monotone. We adopt the following basis model:
\begin{equation}\label{eqn: uniadd non monotone}
    h_l(x) = \sum_{j=1}^{N_l} \beta_j^{(l)} \phi_j^{(l)}(x), \quad \quad l \notin \cI, \quad \quad x \in [0,1],
\end{equation}
where $\{\phi_j^{(l)}(x)\}_{j=1}^{N_l}$ 
are the ``hat'' basis functions $\phi^{(l)}_j(x) = \max\{1 - |(N_l - 1)x - (j-1)|,0\}$.
The basis weights follow the half-normal priors $\beta^{(l)}_j \overset{\text{indep.}}{\sim} \mathcal{N}^+(0, \tau_l^2)$, where $\tau^2_l>0$ is a variance hyperparameter. Such hat basis functions and priors have been broadly used in unconstrained basis modeling \cite{de2012adaptive,maatouk2017gaussian, sanz2022finite, kim2025adaptive}; alternate choices of basis and priors can also be used in \eqref{eqn: uniadd non monotone} depending on modeling preference. Here, no intercept terms are needed for $h_l$ in \eqref{eqn: uniadd} and \eqref{eqn: uniadd non monotone}, since a global intercept is included for the link function model below.

Finally, we adopt the following basis model on the monotone link function $g$:
\begin{equation}\label{eqn: monotone link function}
    g(z)= {\alpha_0}+\sum_{j=1}^{N_g}{\alpha_j}{\psi_j^{(g)}}(z),
\end{equation}
where $\{{\psi_j^{(g)}}\}_{j=1}^{N_g}$ are the same basis functions used in \eqref{eqn: psi definition} for monotone modeling. Similar spike-and-slab priors are assigned on basis weights, i.e.:
\begin{equation}
    \alpha_j = I_{j}^{(g)} Z_j^{(g)}, \quad I_{j}^{(g)} \overset{\text{indep.}}{\sim} \text{Bern}\Big(\frac{\zeta_g}{N_g} \Big), \quad Z_j^{(g)} \overset{\text{indep.}}{\sim} \mathcal{N}^+(0, \tau_g^2), \quad j = 1, \cdots, N_g.
    \label{eqn: weightprior2}
\end{equation}
The intercept $\alpha_0$ is assigned the prior $\alpha_0 \sim \mathcal{N}(\mu_0, \tau_0^2)$. Note that the domain of $g$ in \eqref{eqn: monotone link function} can technically extend beyond $[0,1]$, although $g$ will be flat outside of this interval. For effective non-additive modeling, $\sum_{l=1}^d h_l(x_l)$ should thus be concentrated on $[0,1]$; we ensure this with a careful specification of the variance hyperparameters $\{\tau_l^2\}_{l=1}^d$ (see Section \ref{ssec: alg}).

Given the above model specification for TAIM, one can sample from the posterior distribution of $f$ via MCMC sampling on its model parameters. We present later in Section \ref{ssec: alg} an efficient MCMC algorithm that exploits closed-form updates within a Gibbs sampler. Before this, we first prove and investigate a posterior contraction rate for TAIM in the specific case where all inputs are in the monotone set, i.e., $\cI= \{1, \cdots, d\}$.

\subsection{Posterior Contraction Rates}

\label{sec: theory}

We first define some notation. Let $f_0$ be the true underlying function, which we presume lies in a function space $\mathcal{F}$. Following Section 1.5.1 of \cite{van2011information}, for a given choice of $\alpha = \lceil \alpha \rceil  - 1 + \eta > 0$, define $C^\alpha[0,1]$ as the $\alpha$-H\"{o}lder function space, consisting of functions on $[0,1]$ that are $(\lceil \alpha \rceil - 1)$-times continuously differentiable and whose highest-order derivatives are H\"{o}lder continuous with exponent $\eta \in (0,1]$. Let $\mathcal{D}_n = \{(\mathbf{x}_i,y_i)\}_{i=1}^n$ denote the training data, and let $\bbE_{f_0}(\cdot)$ denote the expectation under the sampling model \eqref{eqn: data} given the true function $f_0$, with variance $\sigma^2_{\epsilon}$ known and design points fixed. Finally, let $\Pi( \cdot | \mathcal{D}_n)$ denote the posterior distribution of $f$ under a considered Bayesian model, e.g., the TAIM model.

Next, define the empirical $L_2$-norm between two functions $f$ and $f_0$ as:
\begin{equation}\label{eqn: empirical L2 norm}
    \|f - f_0\|_n^2 = \frac{1}{n} \sum_{i=1}^n \big(f(\mathbf{x}_i) - f_0(\mathbf{x}_i) \big)^2.
\end{equation}
This norm is widely used for posterior contraction rates in the Bayesian nonparametrics literature \citep{ghosal2007convergence, chakraborty2021convergence, zhou2024mass}, thus we adopt this in our analysis. Under i.i.d. design points, Section 5 of \cite{geer2000empirical} shows this norm is directly related to the functional $L_2$-norm between $f$ and $f_0$. 

In what follows, we wish to prove \textit{posterior contraction rates} of the following form. For a given function $f_0 \in \mathcal{F}$, we wish to find a $K > 0$ such that: 
\begin{equation}\label{eqn: definition of posterior contraction rate - emp L2 norm}
    \lim_{n \rightarrow \infty} \bbE_{f_0}\Pi \big(\|f - f_0\|_n > K\epsilon_n \mid \mathcal{D}_n \big) = 0.
\end{equation}
If the above holds, $\epsilon_n$ is called a \textit{posterior contraction rate} in the Bayesian nonparametrics literature \cite{van2008rates}. Equation \eqref{eqn: definition of posterior contraction rate - emp L2 norm} states that, as $n$ grows, the posterior distribution of $f$ concentrates on a shrinking ball of radius $K\epsilon_n$ around the true function $f_0$. The rate $\epsilon_n$ thus quantifies how quickly the posterior concentrates around $f_0$. We first prove a posterior contraction rate in the one-dimensional (1-d) setting, which shows our 1-d spike-and-slab basis model \eqref{eqn: monotone link function} and \eqref{eqn: weightprior2} achieves the 1-d minimax rate for isotonic regression. We then prove a posterior contraction rate for TAIM in the multi-dimensional setting, which shows that our model softens the curse-of-dimensionality faced by existing contraction rates.

\subsubsection{1-d Setting}
Consider first the prediction of a 1-d monotone function $g_0(x)$. Suppose we adopt on $g(x)$ the basis model \eqref{eqn: monotone link function} (with domain set as $[0,1]$), with discretization level $N$ and the spike-and-slab basis weight priors \eqref{eqn: weightprior2} with fixed hyperparameters. We can show the following 1-d posterior contraction rate:
\begin{theorem}
Suppose the true function $g_0$ is monotone, with $g_0 \in C^\alpha[0,1]$ for some $\alpha > 0$. Let $g$ follow the basis model \eqref{eqn: monotone link function} with discretization level $N$ and priors \eqref{eqn: weightprior2}, where $\log N = \cO(\log n)$. If $N \ge C_1 \epsilon_n^{-1/\min(\alpha, 1)}$ for some constant $C_1> 0$, there exists a $K>0$ such that:
  \begin{equation}
        \lim_{n \rightarrow \infty} \bbE_{g_0}\Pi \big(\|g - g_0\|_n > K\epsilon_n \mid \Data \big) = 0,
    \label{eqn: 1drate}
    \end{equation}
\label{thm: 1d}
for $\epsilon_n = n^{-1/3} (\log n)^{1/3}$.
\end{theorem}
\noindent In other words, given that the true function $g_0 \in C^\alpha[0,1]$ is monotone, then if one uses a discretization level of $N = \lceil C_1 \epsilon_n^{-1/\min(\alpha, 1)} \rceil$ with $\epsilon_n = n^{-1/3} (\log n)^{1/3}$, the posterior distribution of $g$ under our 1-d basis model contracts to $g_0$ at a rate of $\mathcal{O}\{n^{-1/3} (\log n)^{1/3}\}$. Note that the constant $C_1$ can be user-specified; varying this does not affect the contraction rate. Here, $\log N = \cO(\log n)$ is a mild regularity condition that can always be satisfied in practice via an upper bound on the discretization level $N$. The proof is in Appendix \ref{ssec: 1D rate proof}.

In the frequentist literature, it is known \cite{mammen2007additive} that the minimax rate for 1-d isotonic regression (in empirical $L_2$-norm) is $\mathcal{O}(n^{-1/3})$; the above posterior contraction rate thus achieves this up to log terms. In the Bayesian literature, \cite{chakraborty2021convergence} and \cite{zhou2024mass} proved similar posterior contraction rates for a 1-d projection-based and basis monotone model, respectively. The latter rate (which aligns with our basis approach) requires that $g_0$ be monotone with Lipschitz continuous derivatives, followed by a flat region. Using our spike-and-slab basis model, the function class on $g_0$ can be generalized to the $\alpha$-H\"{o}lder space of monotone functions, which includes non-differentiable functions when $\alpha \le 1$. The modeling intuition is as follows. A common strategy for proving contraction rates \cite{van2008rates} is to show that the ``small ball probability'' (i.e., of a sup-norm ball of radius $\epsilon > 0$) around $g_0$ is sufficiently large under the considered prior on $g$. For a monotone $g_0 \in C^\alpha[0,1]$, we can show (see Appendix \ref{ssec: small ball probability}) that this small ball probability requirement is satisfied using only $\mathcal{O}(\epsilon^{-1})$ active basis functions of the form \eqref{eqn: monotone link function}. Our basis model exploits this inherent sparsity via spike-and-slab priors on basis weights with the number of basis growing as $N \geq C_1 \epsilon^{-1/\min(\alpha, 1)}$, which allows us to prove minimax contraction rates for monotone $g_0 \in C^\alpha[0,1]$.

\subsubsection{Multi-dimensional Setting}

Consider next the multi-dimensional setting of $d > 1$. To simplify theory, we consider the special case where all inputs are in the monotone set, i.e., $\mathcal{I} = \{1, \cdots, d\}$. Here, we presume the true function $f_0$ takes the transformed additive isotonic form \eqref{eqn: generalized additive model}, i.e.:
\begin{equation}
f_0(\mathbf{x}) = g_0\left\{\sum_{l=1}^d h_{0,l}(x_l)\right\},
\label{eqn:gam_true}
\end{equation}
where $g_0$ and $\{h_{0,l}\}_{l = 1}^d$ are monotone functions. We adopt on $f$ the TAIM model in Section \ref{sec: taim} with fixed prior hyperparameters. We can then show the posterior contraction rate below:

\begin{theorem}\label{thm: TrAM posterior contraction rate}
    Suppose the true function $f_0$ takes the form \eqref{eqn:gam_true} with $g_0$ and $\{h_{0,l}\}_{l=1}^d$ monotone. Further suppose that $g_0$ is Lipschitz continuous and $h_{0, l} \in C^{\alpha_l}[0, 1]$ for some $\alpha_l > 0$, $l = 1, \cdots, d$. Let $f$ follow the TAIM model \eqref{eqn: generalized additive model}--\eqref{eqn: weightprior2}, where $\log N_g = \cO(\log n)$ and $\log N_l = \cO(\log n)$ for $l = 1, \cdots, d$. If the discretization levels satisfy $N_g \ge C_g (n/\log n)^{1/3}$ and $N_l \geq C_l (n/\log n)^{1/\{3 \min(\alpha_l, 1)\}}$ for some constants $C_g > 0$, $C_l > 0$, $l = 1, \cdots, d$, there exists a $K > 0$ such that:
    \begin{equation}
        \lim_{n \rightarrow \infty} \bbE_{f_0}\Pi \big(\|f - f_0\|_n > K\epsilon_n \mid \Data \big) = 0,
    \end{equation}
    for $\epsilon_n = n^{-1/3} (\log n)^{1/3}$.
\end{theorem}
\noindent In other words, given the true function $f_0$ takes the transformed additive isotonic form \eqref{eqn:gam_true} with mild continuity conditions on its component 1-d functions, then if one uses a discretization level of $N_l = \lceil C_l (n/\log n)^{1/\{3 \min(\alpha_l, 1)\}} \rceil$ in dimension $l$, the posterior distribution of $f$ under our TAIM model contracts to $f_0$ at a rate of $\mathcal{O}\{n^{-1/3} (\log n)^{1/3}\}$. As before, the constants $C_g, C_1, \cdots, C_d$ can be user-specified; varying these does not affect the contraction rate. Here, $\log N_g = \cO(\log n)$ and $\log N_l = \cO(\log n)$ are again mild regularity conditions that can always be satisfied via upper bounds on the discretization levels $N_g, N_1, \cdots, N_d$. The proof of this theorem (see Appendix \ref{ssec: posterior contraction general d}) makes use of the 1-d ``small ball probabilities'' around each function component $g_0$ and $\{h_{0,l}\}_{l=1}^d$ (Lemma \ref{lem: prior probability lower bound} in Appendix), to bound the small ball probability around the true function $f_0$ via Lipschitz continuity on $g_0$.

This theorem reveals several useful insights on our TAIM model. The first is on its predictive performance as dimension $d$ increases. Recall that existing Bayesian isotonic models can suffer from a curse-of-dimensionality, in that their contraction rates can deteriorate rapidly in dimension $d$, e.g., $\mathcal{O}(n^{-1/(2+d)})$ for projection-based models \cite{wang2023posterior}. In contrast, the above contraction rate $\mathcal{O}\{n^{-1/3} (\log n)^{1/3}\}$ does \textit{not} depend on $d$ for terms involving $n$. This suggests that TAIM can indeed ``tame'' this curse-of-dimensionality for posterior contraction in Bayesian isotonic modeling. One reason is that, instead of the full space of monotone functions on $[0,1]^d$, we consider only transformed additive functions of the form \eqref{eqn:gam_true}, which can be effectively modeled in high dimensions via its embedded additive structure. Given the broad use of transformed additive models for flexible high-dimensional modeling, e.g., generalized additive models \cite{hastie2017generalized} and Kolmogorov-Arnold networks \cite{liu2025kan}, this constraint does not appear too restrictive in our later applications. For modeling flexibility, we present later in Section \ref{sec: computation} an extension of TAIM that models for potential deviations of $f$ from transformed additivity.

Another insight can be gleaned by comparing with existing rates for additive isotonic regression. From the frequentist literature, it is known \cite{mammen2007additive} that the minimax rate (in empirical $L_2$-norm) for additive isotonic functions, i.e., of the form $f_0(\mathbf{x}) = \sum_{l=1}^d h_{0,l}(x_l)$, is $\mathcal{O}(n^{-1/3})$. Our rate from Theorem \ref{thm: TrAM posterior contraction rate} thus matches this ``dimension-resistant'' minimax rate up to log factors. The main advantage of TAIM is that it extends this rate to the setting of \textit{transformed} additive isotonic functions, where the unknown link transformation $g$ is estimated from data; such a transformation is key for flexible \textit{non-additive} modeling in high dimensions, as evidenced by the broad use of GAMs and Kolmogorov-Arnold networks.

Finally, note that the discretization levels in Theorem \ref{thm: TrAM posterior contraction rate}, i.e., $N_l = \lceil C_l (n/\log n)^{1/\{3 \min(\alpha_l, 1)\}} \rceil$, depend on the H\"{o}lder orders $\{\alpha_l\}_{l=1}^d$, which control smoothness of the true function $f_0$. These may be known in some applications, e.g., from known smoothness properties of a differential equation solution in computer experiments. When these are unknown, one can set $N_l = \lceil C_l (n/\log n)^{1/\{3 \min(\underline{\alpha}_l, 1)\}} \rceil$, where $\underline{\alpha}_l \leq \alpha_l$ is a conservative lower bound for $\alpha_l$. Using this conservative approach, one still achieves the same rate in Theorem \ref{thm: TrAM posterior contraction rate} (albeit with greater computation), since $N_l \geq C_l (n/\log n)^{1/\{3 \min(\alpha_l, 1)\}}$ is ensured. That said, in practical surrogate modeling problems, the specification of discretization levels is perhaps moreso guided by an allocated computational budget for MCMC sampling. We will thus specify $N_g$ and $\{N_l\}_{l=1}^d$ to satisfy a reasonable MCMC time constraint in later experiments.

\section{Posterior Computation}\label{sec: computation}

Next, we present an efficient posterior sampler for TAIM that scales well in high dimensions. For modeling flexibility, we first introduce an extension of TAIM that models for potential deviations from transformed additivity. We then outline a Gibbs sampling algorithm that exploits closed-form full conditional updates for efficient posterior computation.

\subsection{Transformed Approximate Additive Isotonic Model}
\label{ssec: taaim}
To model for potential deviations of $f$ from transformed additivity, we consider the following transformed \textit{approximate} additive isotonic model (TAAIM):
\begin{equation}\label{eqn: augmented model}
    f(\mathbf{x}) = g\left\{\sum_{l=1}^d h_l(x_l)\right\} + \delta(\mathbf{x}),
\end{equation}
where $\delta(\mathbf{x})$ models this deviation. In our implementation, $\delta$ is modeled by a zero-mean GP prior $\delta \sim \mathrm{GP}(0,\sigma^2_\delta k)$, where $\sigma^2_\delta$ is a variance parameter and $k$ is the squared-exponential kernel $k(\mathbf{x},\mathbf{x}') = \exp\{-\|\mathbf{x} - \mathbf{x}'\|^2/\xi^2\}$. The same priors from Section \ref{sec: taim} are adopted on $g$ and $\{h_l\}_{l=1}^d$. With $\delta \equiv 0$, this reduces to the earlier TAIM model. For a fully Bayesian specification, we further assign priors on model parameters. The GP variance parameter is assigned an Inverse-Gamma prior $\sigma^2_\delta \sim \text{IG}(a_\delta,b_\delta)$. The GP length-scale $\xi$ is assigned a log-uniform prior over $[0.05, 5.0]$. Finally, the noise variance is assigned an Inverse-Gamma prior $\sigma^2_\epsilon \sim \text{IG}(a_\epsilon,b_\epsilon)$. Figure \ref{fig: plate diagram} shows a plate diagram of the full TAAIM model.

\begin{figure}[!t]
\centering
\definecolor{cadd}{HTML}{2C4B8C}
\definecolor{clink}{HTML}{C0392B}
\definecolor{cnoise}{HTML}{7D3C98}
\definecolor{cgp}{HTML}{1E8449}
\begin{tikzpicture}[scale=0.82, transform shape,
    lat/.style={circle, draw, minimum size=10mm, inner sep=0pt, font=\small},
    add/.style={lat, draw=cadd,   fill=cadd!10,   thick},
    lnk/.style={lat, draw=clink,  fill=clink!10,  thick},
    noi/.style={lat, draw=cnoise, fill=cnoise!10, thick},
    gpd/.style={lat, draw=cgp,    fill=cgp!10,    thick},
    det/.style={lat, draw=cadd, densely dashed, thick, fill=white},
    obs/.style={lat, draw=black!75, fill=black!20, thick},
    ed/.style={->, >=stealth, thick, black!65},
    ttl/.style={font=\scriptsize\bfseries, align=center, anchor=south},
    panel/.style={rounded corners=3pt, draw=#1!40, fill=#1!4, line width=0.6pt}]

\path[panel=cadd]  (-1.50,-3.80) rectangle ( 3.25,1.15);
\path[panel=clink] ( 3.55,-3.20) rectangle ( 9.05,1.15);
\path[panel=cnoise]( 9.35,-1.15) rectangle (10.85,1.15);
\path[panel=cgp]   (11.15,-1.15) rectangle (13.45,1.15);
\node[ttl, text=cadd,  text width=4.4cm] at (0.875,1.24) {Basis model for additive\\ components $h_l$};
\node[ttl, text=clink, text width=5.3cm] at (6.300,1.24) {Basis model for link $g$};
\node[ttl, text=cnoise,text width=1.5cm] at (10.10,1.24) {Noise};
\node[ttl, text=cgp,   text width=2.5cm] at (12.30,1.24) {GP discrepancy};

\node[add] (Il) at (0.00, 0.0)  {$I^{(l)}_j$};
\node[add] (Zl) at (1.70, 0.0)  {$Z^{(l)}_j$};
\node[add] (bl) at (0.85,-1.6)  {$\beta^{(l)}_j$};
\draw[ed] (Il) -- (bl);   \draw[ed] (Zl) -- (bl);
\draw[rounded corners] (-0.80,-2.75) rectangle (2.55,0.55);
\node[font=\scriptsize, anchor=south east] at (2.50,-2.71) {$j = 1, \ldots, N_l$};
\draw[rounded corners] (-1.10,-3.35) rectangle (2.85,0.85);
\node[font=\scriptsize, anchor=south west] at (-1.05,-3.31) {$l = 1, \ldots, d$};

\node[lnk] (Ig) at (4.60, 0.0)  {$I^{(g)}_j$};
\node[lnk] (Zg) at (6.30, 0.0)  {$Z^{(g)}_j$};
\node[lnk] (ag) at (5.45,-1.6)  {$\alpha_j$};
\node[lnk] (a0) at (8.25,-1.6)  {$\alpha_0$};
\draw[ed] (Ig) -- (ag);   \draw[ed] (Zg) -- (ag);
\draw[rounded corners] (3.80,-2.75) rectangle (7.15,0.55);
\node[font=\scriptsize, anchor=south east] at (7.10,-2.71) {$j = 1, \ldots, N_g$};

\node[noi] (s2e) at (10.10, 0.0) {$\sigma^2_\epsilon$};
\node[gpd] (s2d) at (11.75, 0.0) {$\sigma^2_\delta$};
\node[gpd] (xi)  at (12.85, 0.0) {$\xi$};

\node[det] (ui) at (1.70,-4.75) {$u_i$};
\node[gpd, minimum size=13mm, font=\scriptsize] (di) at (11.85,-4.75) {$\delta(\mathbf{x}_i)$};
\node[obs] (yi) at (6.30,-6.35) {$y_i$};
\draw[rounded corners] (0.60,-7.40) rectangle (12.90,-4.00);
\node[font=\scriptsize, anchor=south east] at (12.85,-7.36) {$i = 1, \ldots, n$};
\draw[ed] (bl)  -- (ui);   \draw[ed] (ui)  -- (yi);
\draw[ed] (ag)  -- (yi);   \draw[ed] (a0)  -- (yi);
\draw[ed] (s2e) to[out=-92,in=-18] (yi);   \draw[ed] (di)  -- (yi);
\draw[ed] (s2d) -- (di);   \draw[ed] (xi)  -- (di);
\end{tikzpicture}
\caption{Plate diagram of the TAAIM model in Section \ref{ssec: taaim}. Shaded colored boxes group different parts of the model: additive components in blue, link function in red, noise in purple, and GP discrepancy in green. Solid circles mark the parameters to be sampled via MCMC, the dashed circle marks the relation $u_i = \sum_{l=1}^d h_l(x_{il})$, and the bottom shaded circle marks the observed responses.}
\label{fig: plate diagram}
\end{figure}
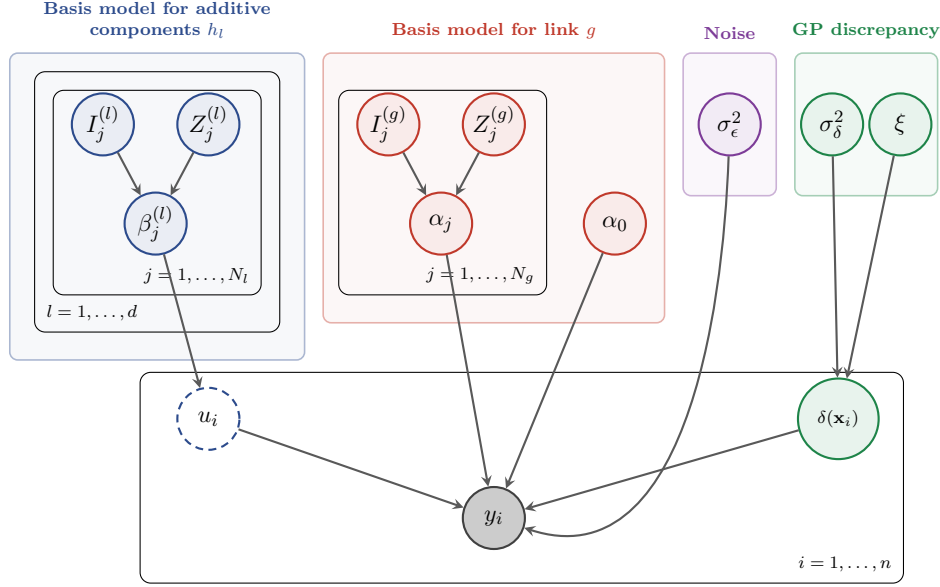

One potential issue with the above model is that the isotonicity property \eqref{eq:monotone} is no longer guaranteed almost surely, since the GP discrepancy $\delta$ does not ensure isotonicity. This can be practically addressed in two ways. First, one can use a careful specification of $(a_\delta,b_\delta)$ (see Section \ref{ssec: alg}) to ensure that $\sigma^2_\delta$ is small, which encourages a small discrepancy $\delta$ unless strongly suggested by the data. Second, one can perform a post-hoc check on the fitted model to see if isotonicity is violated, e.g., using the nonparametric hypothesis tests in \cite{scott2015nonparametric}. If isotonicity is violated, one can refit the model with a smaller prior distribution on $\sigma^2_\delta$.

We present next a posterior sampler for this extended TAAIM model. Such a sampler can naturally be used for the earlier TAIM model by simply ignoring sampling steps involving $\delta$ (see Section \ref{ssec: alg} for details). In practice, a hierarchical model selection approach (e.g., DIC \cite{spiegelhalter2002bayesian}) can be used to choose between TAIM and TAAIM if needed.

\subsection{Gibbs Sampling Steps}\label{ssec: gibbs}

Let $\boldsymbol{\Theta}= \{\boldsymbol{\Theta}_g, \boldsymbol{\Theta}_h, \boldsymbol{\Theta}_\delta,\sigma^2_{\epsilon}\}$ denote the model parameters in TAAIM to infer. Here, $\boldsymbol{\Theta}_g = \{ \alpha_0, I_1^{(g)}, \cdots, I_{N_g}^{(g)}, Z_1^{(g)}, \cdots, Z_{N_g}^{(g)}\}$ are the parameters for the link function $g$, $\boldsymbol{\Theta}_{h} = \{ \{I_1^{(l)}, \cdots, \allowbreak I_{N_l}^{(l)}\}_{l \in \mathcal{I}}, \{Z_1^{(l)}, \cdots, Z_{N_l}^{(l)}\}_{l=1}^d\}$ are the parameters for the additive components $\{h_l\}_{l=1}^d$, and $\boldsymbol{\Theta}_\delta = \{\boldsymbol{\delta},\sigma^2_\delta,\xi\}$ are for the GP discrepancy, where $\bdelta = [\delta(\mathbf{x}_i)]_{i=1}^n$ are the latent values of $\delta$ at design points\footnote{For simplicity, we assume no repeated design points; our sampler extends naturally with repeated points.}. Conditional on data $\mathcal{D}_n = \{(\mathbf{x}_i,y_i)\}_{i=1}^n$, we show next that the full conditional distributions of each parameter in $\boldsymbol{\Theta}$ (except the length-scale $\xi$) can be obtained in closed form, enabling efficient Gibbs sampling. Derivation details are in Appendix \ref{ssec: full conditional derivation}.

Consider first $\boldsymbol{\Theta}_g$. Let $\boldsymbol{\Theta}_{-}$ denote the parameter set $\boldsymbol{\Theta}$ with the considered parameter omitted. Adopting the convention $\psi_0^{(g)} \equiv 1$, define the residual term of the $i$-th data point excluding the contribution from the $j$-th basis term of $g$ as:
\begin{equation}\label{eqn: partial residual}
    r_{ij} = y_i - \delta(\mathbf{x}_i) - \sum_{j' \neq j} \alpha_{j'} \psi_{j'}^{(g)}(u_i), \quad i = 1, \cdots, n, \quad j = 0, \cdots, N_g,
\end{equation}
where $u_i = \sum_{l=1}^d h_l(x_{il})$ is the link function argument at $\mathbf{x}_i = (x_{i1}, \cdots, x_{id})$. One can then show that the full conditional distribution of the intercept $\alpha_0$ is normally-distributed:
\begin{equation}
[\alpha_0 | \boldsymbol{\Theta}_{-}, \mathcal{D}_n] \sim \mathcal{N}(\widehat{\mu}_0, \widehat{\tau}_0^2), \quad
\widehat{\tau}_0^2 = \Big( \frac{1}{\tau_0^2} + \frac{n}{\sigma^2_\epsilon} \Big)^{-1}, \quad
\widehat{\mu}_0 = \widehat{\tau}_0^2 \Big( \frac{\mu_0}{\tau_0^2} + \frac{1}{\sigma^2_\epsilon} \sum_{i=1}^n r_{i0} \Big).
\label{eqn: fullcondint}
\end{equation}
The full conditional distributions of the indicators $I_1^{(g)}, \cdots, I_{N_g}^{(g)}$ take the Bernoulli form:
\vspace{-0.4cm}
\small
\begin{equation}
[I_j^{(g)} | \boldsymbol{\Theta}_{-}, \mathcal{D}_n] \sim \text{Bern}(\widehat{p}_j), \;
\log \frac{\widehat{p}_j}{1 - \widehat{p}_j} = \log \frac{\zeta_g}{N_g - \zeta_g} + \frac{Z_j^{(g)}}{\sigma^2_\epsilon} \sum_{i=1}^n \psi_j^{(g)}(u_i) \left\{ r_{ij} - \frac{Z_j^{(g)}}{2} \psi_j^{(g)}(u_i) \right\}.
\label{eqn: fullcondind1}
\vspace{-0.4cm}
\end{equation}
\normalsize
\hspace{-0.68cm} Finally, the full conditionals of the magnitudes $Z_1^{(g)}, \cdots, Z_{N_g}^{(g)}$ are truncated normal:
\begin{equation}
[Z_j^{(g)} | \boldsymbol{\Theta}_{-}, \mathcal{D}_n] \sim
\begin{cases}
\mathcal{N}^+(\widehat{\mu}_j, \widehat{\tau}_j^2), & I_j^{(g)} = 1, \\
\mathcal{N}^+(0, \tau_g^2), & I_j^{(g)} = 0,
\end{cases}
\label{eqn: fullcondmag1}
\end{equation}
where $\widehat{\mu}_j = \widehat{\tau}_j^2 \sigma^{-2}_\epsilon \sum_{i=1}^n \psi_j^{(g)}(u_i) r_{ij}$ and $\widehat{\tau}_j^2 = \Big\{ {\tau_g^{-2}} + {\sigma^{-2}_\epsilon} \sum_{i=1}^n \psi_j^{(g)}(u_i)^2 \Big\}^{-1}$. When $I_j^{(g)} = 0$, i.e., the $j$-th basis function for $g$ is inactive, the magnitude $Z_j^{(g)}$ does not enter the likelihood, thus its full conditional distribution defaults to its prior distribution.

Consider next $\boldsymbol{\Theta}_h$. For a monotone input variable $l \in \mathcal{I}$, we sample its basis indicator $I_j^{(l)}$ and magnitude $Z_j^{(l)}$ as a block via the full conditional distribution $[(I_j^{(l)},Z_j^{(l)})|\tilde{\boldsymbol{\Theta}}_{-},\mathcal{D}_n]$, where $\tilde{\boldsymbol{\Theta}}_{-}$ is the parameter set $\boldsymbol{\Theta}$ without $I_j^{(l)}$ and $Z_j^{(l)}$. This can be done by first sampling $[I_j^{(l)}|\tilde{\boldsymbol{\Theta}}_{-},\mathcal{D}_n]$, then sampling $[Z_j^{(l)}|I_j^{(l)},\tilde{\boldsymbol{\Theta}}_{-},\mathcal{D}_n]$. We can show that:
\begin{equation}
[I_j^{(l)} \mid \tilde{\boldsymbol{\Theta}}_{-}, \mathcal{D}_n] \sim \text{Bern}(\widehat{p}^{(l)}_j), \quad
\log \frac{\widehat{p}^{(l)}_j}{1-\widehat{p}^{(l)}_j}
= \log \frac{\zeta_l}{N_l - \zeta_l}
+ \log \Big\{ \frac{2}{\tau_l} \sum_{v=0}^M \omega_v \Big\} + c_{v_0},
\label{eqn: fullcondind2}
\end{equation}
where $\omega_v$ and $c_v$ are coefficients and $v_0$ is an index, all defined in Appendix \ref{ssec: full conditional derivation} for brevity. The sampling of $[Z_j^{(l)}|I_j^{(l)},\tilde{\boldsymbol{\Theta}}_{-},\mathcal{D}_n]$ is more involved. Note that when $I_j^{(l)} = 1$, i.e., the $j$-th basis function for $h_l$ is active, the link function argument $u_i$ is affine in $Z_j^{(l)}$. Exploiting this along with the piecewise linear nature of $g$, we can show (details in Appendix \ref{ssec: full conditional derivation}) that:
\begin{equation}
[Z_j^{(l)} \mid I_j^{(l)}, \tilde{\boldsymbol{\Theta}}_{-}, \mathcal{D}_n] \sim
\begin{cases}
\displaystyle \sum_{v=0}^{M} \frac{\omega_v}{\sum_{v'=0}^M \omega_{v'}}\, \mathcal{N}_{[L_v,R_v]}(\widehat{\mu}_v, \widehat{\sigma}_v^2), & I_j^{(l)} = 1, \\[2.2ex]
\mathcal{N}^+(0,\tau_l^2), & I_j^{(l)} = 0,
\end{cases} \quad \quad \text{ for } l \in \mathcal{I}.
\label{eqn: fullcondmag2}
\end{equation}
Here, $\mathcal{N}_{[a,b]}(\mu,\sigma^2)$ denotes the normal distribution $\mathcal{N}(\mu,\sigma^2)$ truncated to $[a,b]$, and $M < \infty$, $\widehat{\sigma}_v^2$, $\widehat\mu_v$, $L_v$ and $R_v$ are defined in Appendix \ref{ssec: full conditional derivation} for brevity. In other words, when $I_j^{(l)} = 1$, the full conditional of $Z_j^{(l)}$ is a finite mixture of truncated 1-d normal distributions, which can be efficiently sampled. When $I_j^{(l)} = 0$, its full conditional again defaults to the prior. For a non-monotone input variable $l \notin \mathcal{I}$, one can use a similar rationale to show that:
\begin{equation}
[Z_j^{(l)} \mid \boldsymbol{\Theta}_{-}, \mathcal{D}_n] \sim \sum_{v=0}^{M} \frac{\omega_v}{\sum_{v'=0}^M \omega_{v'}}\, \mathcal{N}_{[L_v,R_v]}(\widehat{\mu}_v, \widehat{\sigma}_v^2), \quad \text{ for } l \notin \mathcal{I}.
\label{eqn: fullcondmag3}
\end{equation}
In practice, we find that an ``on-the-fly'' normalization (used in spatial statistics; see, e.g., \cite{rue2005gaussian,mak2016regional}) of the sampled basis weights $\{\beta_j^{(l)}\}$ can improve parameter identifiability and MCMC mixing. Appendix \ref{ssec: onthefly} provides further details on this step.

Finally, consider $\boldsymbol{\Theta}_\delta$. Define $\mathbf{K} = [k(\mathbf{x}_i,\mathbf{x}_{i'})]_{i,i'=1}^n$, $\mathbf{M} = \sigma^2_\delta \mathbf{K} + \sigma^2_\epsilon \mathbf{I}_n$, and let $\mathbf{r} = (r_1,\cdots,r_n)^T$, $r_i = y_i - g(u_i)$ be the current residual vector. One can show that the full conditional distribution of the latent vector $\boldsymbol{\delta}$ is multivariate normal:
\begin{equation}
[\boldsymbol{\delta} \mid \boldsymbol{\Theta}_{-}, \mathcal{D}_n] \sim
\mathcal{N}\big( \sigma^2_\delta\, \mathbf{K}\mathbf{M}^{-1} \mathbf{r}, \sigma^2_\epsilon \sigma^2_\delta\, \mathbf{K}\mathbf{M}^{-1} \big),
\label{eqn: fullconddelta}
\end{equation}
and the full conditional distribution of the GP variance $\sigma^2_\delta$ is Inverse-Gamma:
\begin{equation}
[\sigma^2_\delta \mid \boldsymbol{\Theta}_{-}, \mathcal{D}_n] \sim
\text{IG}\Big( a_\delta + \frac{n}{2}, \;\; b_\delta + \frac{1}{2}\boldsymbol{\delta}^T \mathbf{K}^{-1}\boldsymbol{\delta}\Big).
\label{eqn: fullcondgpvar}
\end{equation}
Finally, the full conditional distribution of the noise variance $\sigma^2_{\epsilon}$ is also Inverse-Gamma:
\begin{equation}
[\sigma^2_\epsilon \mid \boldsymbol{\Theta}_{-}, \mathcal{D}_n] \sim
\text{IG}\Big( a_\epsilon + \frac{n}{2}, \;\; b_\epsilon + \frac{1}{2}\sum_{i=1}^n \big\{y_i - g(u_i) - \delta(\mathbf{x}_i)\big\}^2 \Big),
\label{eqn: fullcondnoisevar}
\end{equation}
and the GP length-scale $\xi$ is sampled using the Metropolis-Hastings algorithm \cite{hastings1970monte}. 

\begin{algorithm}[t]
\caption{Gibbs sampling algorithm for TAAIM}\label{alg: gibbs}
\begin{algorithmic}[1]
\STATE \textbf{Inputs:} Data $\{(\mathbf{x}_i, y_i)\}_{i=1}^n$, MCMC iterations $T_{\rm iter}$, discretization levels $N_g$ and $\{N_l\}_{l=1}^d$, indicator flag \texttt{taaim.flg}
\STATE $\bullet$ Set initial values for parameters $\boldsymbol{\Theta}= \{\boldsymbol{\Theta}_g, \boldsymbol{\Theta}_h, \boldsymbol{\Theta}_\delta,\sigma^2_{\epsilon}\}$, with $\boldsymbol{\delta} = 0$.
\FOR{$t = 1, \ldots, T_{\rm iter}$}
    \STATE $\bullet$ Sample $\alpha_0$ from the normal distribution \eqref{eqn: fullcondint}.
    \STATE $\bullet$ Sample each $\{I_j^{(g)}\}_{j=1}^{N_g}$ from the Bernoulli distribution \eqref{eqn: fullcondind1}.
    \STATE $\bullet$ Sample each $\{Z_j^{(g)}\}_{j=1}^{N_g}$ from the truncated normal distribution \eqref{eqn: fullcondmag1}.

    \FOR{ $l = 1, \cdots, d$ and $j = 1, \cdots, N_l$ }

    \IF{ $l \in \mathcal{I}$}
    \STATE $\bullet$ Sample $I_j^{(l)}$ from the Bernoulli distribution \eqref{eqn: fullcondind2}.
    \STATE $\bullet$ Sample $Z_j^{(l)}$ from the distribution \eqref{eqn: fullcondmag2}.
    \ELSE
    \STATE $\bullet$ Sample $Z_j^{(l)}$ from the distribution \eqref{eqn: fullcondmag3}.
    \ENDIF
    \ENDFOR
    \STATE $\bullet$ Perform an on-the-fly normalization of weights $\{\beta_j^{(l)}\}$ following Appendix \ref{ssec: onthefly}.

    \STATE $\bullet$ Sample $\sigma^2_\epsilon$ from the Inverse-Gamma distribution \eqref{eqn: fullcondnoisevar}.

        \IF{ \texttt{taaim.flg = true} }
        \STATE $\bullet$ Sample $\boldsymbol{\delta}$ from the multivariate normal distribution \eqref{eqn: fullconddelta}.
        \STATE $\bullet$ Sample $\sigma^2_\delta$ from the Inverse-Gamma distribution \eqref{eqn: fullcondgpvar}.
        \STATE $\bullet$ Sample $\xi$ via Metropolis-Hastings.
    \ENDIF
\ENDFOR
\end{algorithmic}
\end{algorithm}

\subsection{Sampling Algorithm and Computational Complexity}\label{ssec: alg}

Algorithm \ref{alg: gibbs} outlines the full Gibbs sampling algorithm, where \texttt{taaim.flg} indicates whether TAAIM (Section \ref{ssec: taaim}) is used instead of TAIM (Section \ref{sec: taim}). We first initialize the parameters in $\boldsymbol{\Theta}$, where the latent vector $\boldsymbol{\delta}$ is initialized at $\boldsymbol{\delta}=\boldsymbol{0}$. Each sampling step from Section \ref{ssec: gibbs} is then performed iteratively. A burn-in period and MCMC diagnostics \cite{gelman1992inference} should be used to ensure good mixing. We find that an effective strategy is to first sample only $\boldsymbol{\Theta}_g$, $\boldsymbol{\Theta}_h$ and $\sigma^2_{\epsilon}$ in some warm-up iterations at the start of the burn-in, then add in the sampling of $\bTheta_\delta$ after. This allows the sampler to explore the transformed additive structure of $f$ before accounting for the GP discrepancy, which seems to improve MCMC mixing.

The specification of prior hyperparameters is also important. For the spike-and-slab priors, we set $\zeta_g$ and $\{\zeta_l\}_{l \in \mathcal{I}}$ as 5, which seems to work well in experiments. The variance hyperparameters $\tau^2_1, \cdots, \tau^2_d$ are set using the second-moment approach in \cite{chipman2010bart} to ensure that $\sum_{l=1}^d h_l(x_l)$ is concentrated on $[0,1]$, as desired from the discussion in Section \ref{ssec: model}; Appendix \ref{ssec: data adaptive hyperparameters} provides details. We further find that a data-dependent specification of $\mu_0$, $\tau^2_0$ and $\tau^2_g$ (also detailed in Appendix \ref{ssec: data adaptive hyperparameters}) works well in capturing the function scale with limited training data; similar data-dependent priors have been used in Bayesian regression \citep{chipman2010bart, richardson1997bayesian} and GP modeling \cite{chen2023}. For the GP discrepancy, we set $(a_\delta,b_\delta) = (3.0,0.1)$, which encourages a small discrepancy term $\delta$ unless strongly suggested by the data, as desired from the discussion in Section \ref{ssec: taaim}. For the noise model, we set $(a_\epsilon,b_\epsilon) = (2.0,0.1)$.

A key computational advantage of this posterior sampler is its efficiency in high dimensions. A careful complexity analysis of Algorithm \ref{alg: gibbs} shows that each Gibbs sampling iteration requires $\cO \{nN_g(\sum_{l=1}^d N_l)\log n + n^3 \}$ work; Appendix \ref{ssec: complexity} provides a full derivation. When each input uses the same discretization level $N_l = N$, this reduces to 
$\cO\{nN_g N d \log n + n^3\}$ work. Thus, with $n$, $N_g$ and $N$ held fixed, this complexity scales linearly in dimension $d$, which facilitates efficient posterior sampling in high dimensions. This is in contrast to existing Bayesian isotonic models (see Section \ref{sec: background}), where drawing a single posterior sample from the monotone-constrained process may require work growing exponentially in $d$. Given a computational budget, the latter models may thus require a coarse discretization for model fitting, which hurts their predictions. Our model ``tames'' this curse-of-dimensionality by leveraging a transformed additive form with reduced model parameters, coupled with efficient closed-form updates for Gibbs sampling.

Finally, using the posterior samples on $\boldsymbol{\Theta}$, one can draw posterior samples on $f(\mathbf{x}_{\rm new})$ at a new input $\mathbf{x}_{\rm new}$ as follows. First, given the current iterate of the latent vector $\boldsymbol{\delta}$, one can draw a posterior sample on $\delta(\mathbf{x}_{\rm new})$ via the GP predictive distribution $\mathcal{N}(\mu_\delta,s^2_\delta)$, where:
\begin{equation}
    \mu_\delta = \mathbf{k}_*^\top \mathbf{K}^{-1} \boldsymbol{\delta}, \qquad
    s^2_\delta = \sigma^2_\delta \big( 1 - \mathbf{k}_*^\top \mathbf{K}^{-1} \mathbf{k}_* \big),
    \qquad \mathbf{k}_* = [k(\mathbf{x}_{\rm new},\mathbf{x}_i)]_{i=1}^n .
    \label{eq: gppred}
\end{equation}
Next, one can plug in this sample of $\delta(\mathbf{x}_{\rm new})$ within \eqref{eqn: augmented model} to draw a posterior sample on $f(\mathbf{x}_{\rm new})$, where $g$ and $\{h_l\}$ are specified using the current iterate of $\boldsymbol{\Theta}$. One can then repeat these steps over all iterates of $\boldsymbol{\Theta}$ to sample the posterior predictive distribution for $f(\mathbf{x}_{\rm new})$.

\section{Numerical Experiments}\label{sec: simulation}

We now investigate the performance of TAIM and TAAIM in numerical experiments. We consider the following test functions:
\begin{itemize}[leftmargin=*, noitemsep]

\item \textit{2-d synthetic function}: $f_1(\mathbf{x}) = 0.5\exp\{\log m(x_1) + \log m(x_2)\} - 4.5$, where $m(t) = t + \sin(4\pi t)/(4\pi) + 3$. Here, $f_1$ is monotone in both inputs, thus $\mathcal{I} = \{1, 2\}$.
\item \textit{5-d synthetic function}: $f_2(\mathbf{x}) 
    = \max\Big\{3^{-1} \Big[ 2z_1 + 2\sin(4\pi z_1)/(4\pi) + 5\sin(20z_2) + 5\sin(10x_3) 
      -2(0.4 - x_4)_+ + 2(x_4-0.7)_+
      +(x_5^2-0.5)_+ + 10
      \Big], 3\Big\} + 0.5x_2 + x_3$, where $z_1 = \max(x_1,0.3)$, $z_2 = \max(x_2,0.5)$, $(\cdot)_+ = \max(\cdot,0)$. Here, $f_2$ is monotone in $x_1$, $x_4$, $x_5$, thus $\mathcal{I} = \{1, 4, 5\}$.

      \item \textit{9-d synthetic function}: $f_3(\mathbf{x}) = (u/2)^3 + \exp(0.9u)/5 + 0.5\sin(\pi x_2)\cos(\pi x_3)$, where:
      \begin{align*}
        u & = m(x_1) + (x_4 - 0.3)_+/0.7 + x_5^2 + \{1 + e^{-15(x_6 - 0.5)}\}^{-1} + m(x_8) \\
        & 
        \quad \quad      \quad \quad + 0.15\sin(2\pi x_2) + 0.15\cos(2\pi x_3) + 0.1\sin(2\pi x_7) + 0.4\,x_9(1 - x_9).
      \end{align*}
      Here, $f_3$ is monotone in $x_1$, $x_4$, $x_5$, $x_6$, $x_8$, thus $\mathcal{I} = \{1, 4, 5, 6, 8\}$.
      
      \item \textit{8-d borehole function}: This is the borehole test function $f_4$ in \cite{surjanovic2013virtual}, which models the water flow rate through a borehole. From physical knowledge, $f_4$ should be monotone decreasing in inputs $x_2$, $x_6$ and $x_7$ (e.g., borehole length), and monotone increasing in all other inputs (e.g., borehole hydraulic conductivity). 
      \item \textit{10-d wing weight function}: This is the wing weight test function $f_5$ in \cite{surjanovic2013virtual}, which models the weight of a light aircraft wing. From physical knowledge, $f_5$ should be monotone decreasing in input $x_7$ (aerofoil thickness-to-chord ratio), and monotone increasing in all other inputs except $x_4$ (e.g., load factor and wing area).
\end{itemize}

\noindent Note that the first function $f_1$ is in the transformed additive form \eqref{eqn: generalized additive model}, whereas the remaining functions $f_2, \cdots, f_5$ have some deviation from this form. The borehole and wing weight functions are common test functions in the computer experiments literature.

We compare with the same benchmark models in Section \ref{sec: background}. For fair comparison, all models are constrained to 30 minutes for MCMC sampling, which is already quite time-intensive for surrogate modeling. Each model uses at least 2000 MCMC iterations (with the first 1000 discarded as burn-in), and MCMC convergence is checked using the Gelman-Rubin statistic \cite{gelman1992inference}. The compared models include:
\begin{itemize}[leftmargin=*, noitemsep]
\item \textit{Standard GP}: This is the standard GP benchmark using an anisotropic squared-exponential kernel, implemented using the \texttt{DiceKriging} package \cite{R_DiceKriging} with model parameters estimated using maximum likelihood.

\item \textit{Basis-GP}: This is the basis expansion isotonic GP in \cite{maatouk2017gaussian}, implemented using the algorithm in \cite{ray2020efficient}. Here, we used the finest discretization level $N_l = N$ that ensures the MCMC time constraint is satisfied. For $f_1, \cdots, f_5$, this corresponds to $N=10, 4, 2, 2, 2$, respectively, yielding $10^2$, $4^5$, $2^9$, $2^8,2^{10}$ basis functions. While $N=2$ seems quite small for the latter functions, increasing this to $N=3$ results in $3^{8}-3^{10} \approx 6500-60000$ basis functions, which makes MCMC sampling intractable given the time constraint.

\item \textit{Proj-GP}: This is the projection-based isotonic GP in \cite{lin2014bayesian}. A similar discretization scheme is used as Basis-GP to satisfy the MCMC time constraint.

\item \textit{TAIM and TAAIM}: Our models are implemented following Section \ref{ssec: alg}, with a discretization level of $N_g = 10$ and $N_1 = \cdots = N_d = 10$ to satisfy the MCMC time constraint. Here, this finer discretization is possible due to the scalability of our MCMC sampler in higher dimensions. A warm-up period of 200 iterations is used following Section \ref{ssec: alg}.
\end{itemize}
\noindent All models use the same LHD points, with sample size $n=10d$ following the rule-of-thumb in \cite{loeppky2009choosing}. Data are sampled with Gaussian noise with standard deviation 0.5 for $f_1$, 0.1 for $f_2$, 0.3 for $f_3$, and 1.0 for $f_4$ and $f_5$. Each experiment is replicated 20 times.

\begin{figure}[!t]
    \centering
    \includegraphics[width=0.8\linewidth]{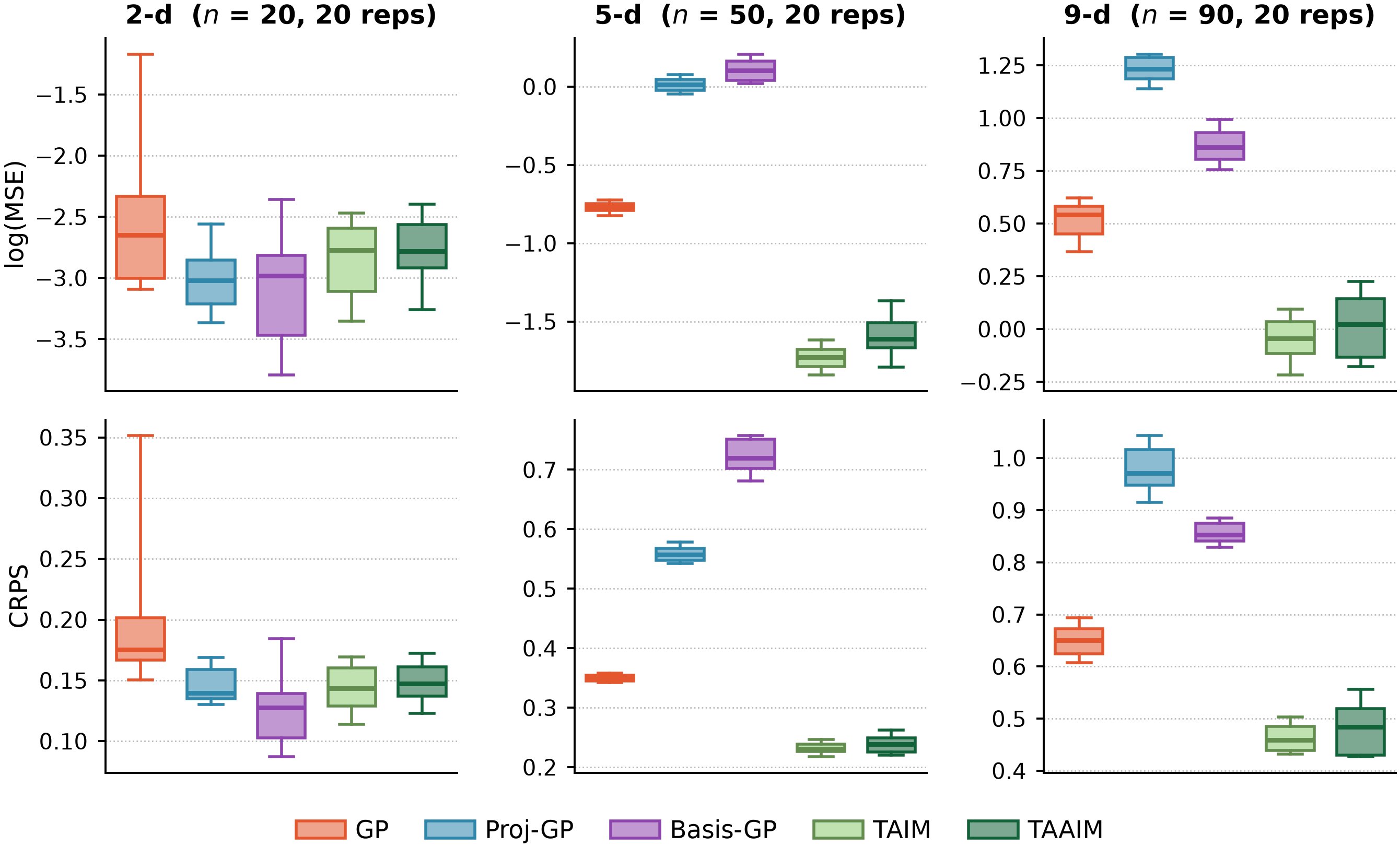}
    \caption{Boxplots of the test log-MSE (top) and CRPS (bottom) for the three synthetic functions $f_1$, $f_2$ and $f_3$, over $20$ replications.}
    \label{fig:Metric synthetic}
\end{figure}

\begin{figure}
    \centering
    \includegraphics[width=0.8\linewidth]{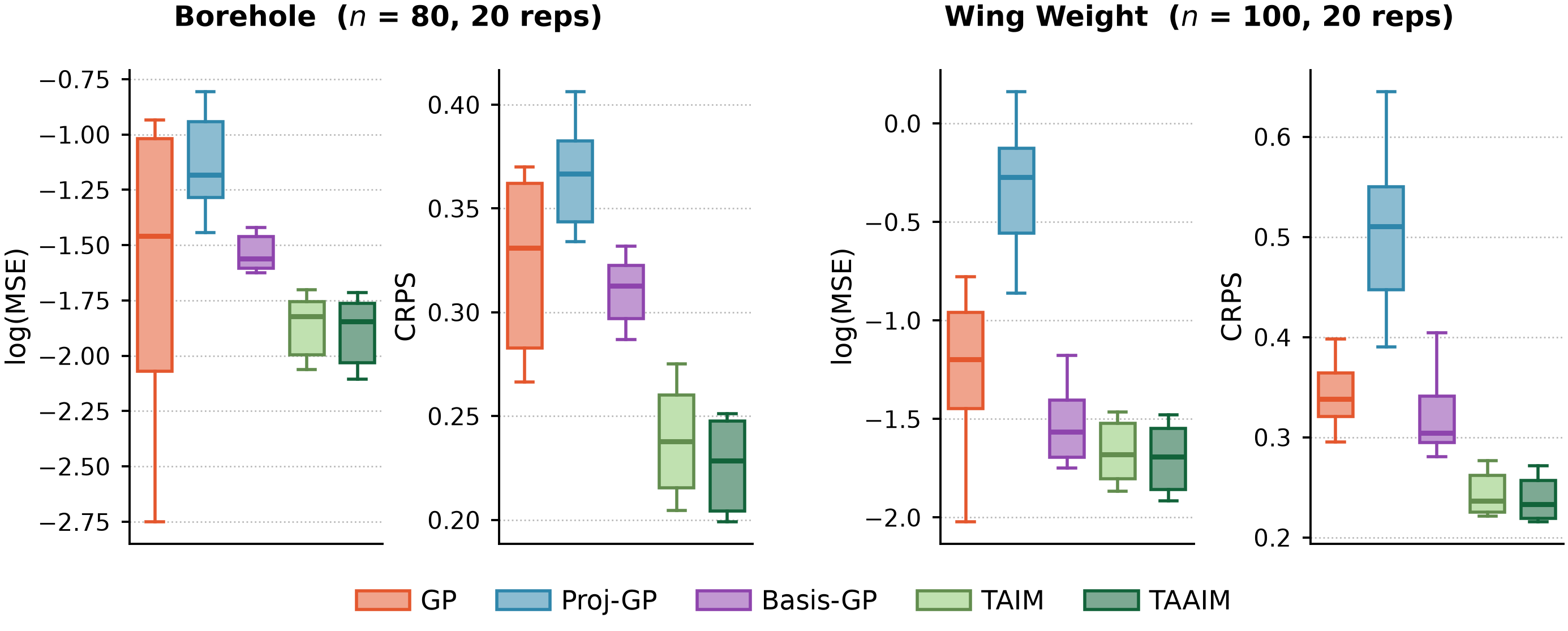}
    \caption{Boxplots of the test log-MSE and CRPS for the borehole function $f_4$ (left two plots) and the wing weight function $f_5$ (right two plots) over $20$ replications. }
    \label{fig:Metric testfun}
\end{figure}

Figure \ref{fig:Metric synthetic} shows the test log-MSE and CRPS boxplots for the three synthetic functions. There are several interesting observations. In low dimensions (i.e., 2-d), the existing isotonic models Proj-GP and Basis-GP perform better than the standard GP, which is intuitive since monotonicity is incorporated and a high-resolution model is computationally feasible. Here, our models perform comparably to these benchmarks. However, in higher dimensions (i.e., 5-d and 9-d), both existing models perform worse than the standard GP despite modeling for monotonicity; this is likely due to the coarse discretization needed for feasible model fitting with many inputs.
In this large-$d$ setting, our models perform considerably better than the benchmarks for both MSE and CRPS,
which shows the value of a transformed additive form for flexible and efficient isotonic modeling in higher dimensions. 

Figure \ref{fig:Metric testfun} shows the same metrics for the borehole and wing weight functions. 
Here, Basis-GP gives only slight improvements over the standard GP for both functions, with Proj-GP performing worse. This again shows that a coarse discretization (which is needed for feasible model fitting) can inhibit the effective incorporation of isotonicity information for modeling. In contrast, our TAIM and TAAIM models perform considerably better than all benchmarks for both functions, which shows that a transformed additive framework (coupled with an efficient Gibbs sampler) can facilitate effective isotonic modeling in higher dimensions. Here, TAAIM performs slightly better than TAIM, which is expected since these two test functions model physical systems that likely have some deviations from transformed additivity. 

\section{Surrogate Modeling Applications}\label{sec: application}

Finally, we consider two applications on the surrogate modeling of expensive simulators.

\subsection{Jet Turbine Deformation}\label{ssec: jet engine}

Consider first the design of jet engines, which are broadly used in commercial aviation. Jet engines operate by sucking in heated gas into a compressor, subjecting it to extreme pressure and temperature, then discharging it through a turbine to generate thrust. Given such extreme conditions, it is thus imperative to design jet turbine blades to reliably withstand deformation in harsh operating conditions.

\begin{figure}[t]
    \centering
    \includegraphics[width=\linewidth]{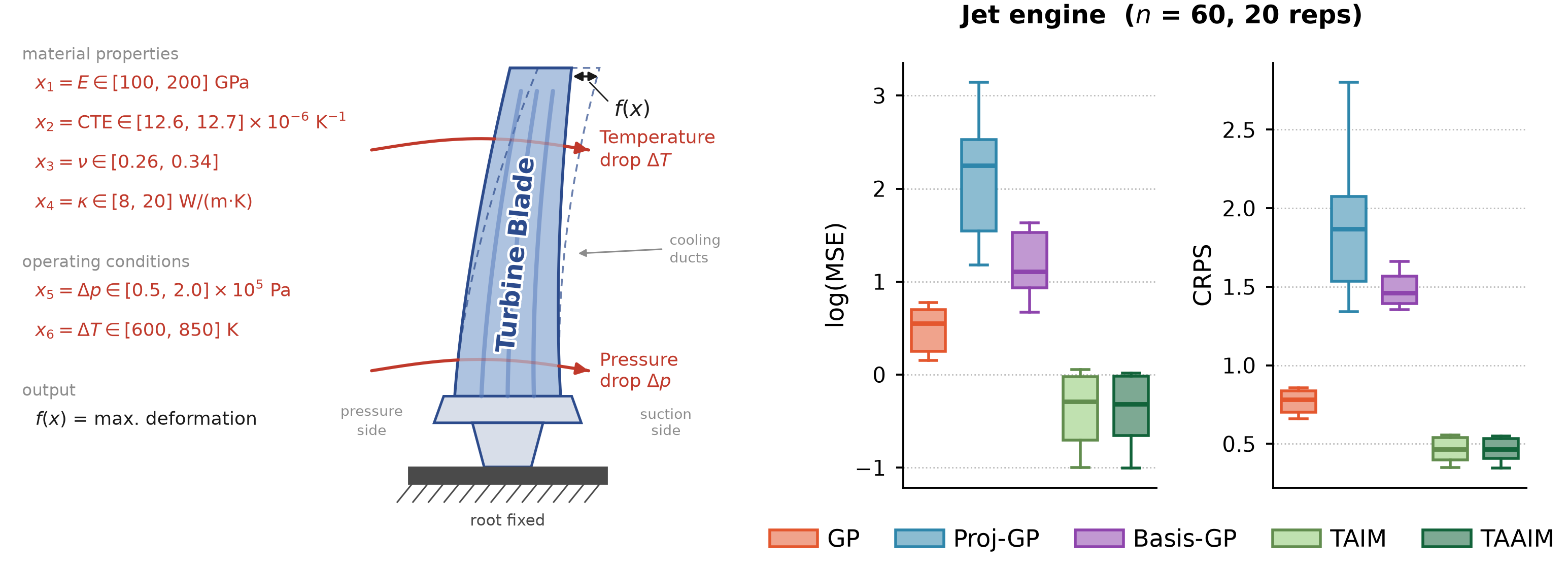}
    \caption{(Left) Visualizing the jet turbine blade simulator set-up, with its $d=6$ inputs and their input ranges. (Right) Boxplots of the test log-MSE and CRPS for the jet turbine blade application over 20 replications.}
    \label{fig:jet-engine}
\end{figure}

Figure \ref{fig:jet-engine} (left) shows the simulation set-up for our jet turbine blade, implemented in the MATLAB module in \cite{matlab_jet}. We investigate $d=6$ inputs: the first four inputs $x_1, \cdots, x_4$ (Young's modulus, coefficient of thermal expansion, Poisson's ratio, thermal conductivity) dictate material properties of the turbine blade, and the last two inputs $x_5, x_6$ (pressure drop, temperature drop) control operating conditions. Given such inputs, the simulator uses a complex static-structural thermal expansion model to simulate the maximum thermal deformation over the blade. Here, the simulator has stochastic noise. The goal is to train a surrogate model that efficiently predicts (with UQ) the mean output from this expensive simulator over an input range (Figure \ref{fig:jet-engine} left), which can then be used to guide the material design of turbine blades.

A careful analysis of this physical system reveals useful isotonicity properties that can be incorporated. First, on operating conditions, as pressure drop ($x_5$) and temperature drop ($x_6$) increase, the amount of blade deformation should naturally increase. Second, on material properties, 
as the coefficient of thermal expansion $x_2$ (which quantifies how much a material expands when heated) increases, or as Young's modulus $x_1$ (which measures a material's stiffness) decreases, one also expects the deformation of the blade to naturally increase. We will investigate how effective the compared models are at incorporating such isotonicity information for surrogate modeling.

The experiment is structured as follows. We use the same benchmark models and set-up from Section \ref{sec: simulation}, subject to the same 30-minute time constraint for MCMC. Proj-GP and Basis-GP are run at the finest discretization level ($N=3$) that ensures this time constraint; this yields $3^6 = 729$ basis functions for the latter model. All models use the same LHD points with a sample size of $n=10d=60$. Each experiment is replicated $20$ times. 

Figure \ref{fig:jet-engine} (right)
shows the test log-MSE and CRPS boxplots for this jet turbine application. Both existing isotonic models (Proj-GP and Basis-GP) perform considerably worse than the standard GP benchmark, which suggests that the high-dimensional challenges faced by such models can outweigh the incorporation of isotonicity information. In contrast, our models perform considerably better than the benchmarks, which again shows the value of a flexible transformed additive modeling form for effective isotonic modeling with many input variables. Here, TAIM and TAAIM yield similar performance.

\subsection{Aircraft Wing Deflection}\label{ssec: wing spar}

Consider next the load analysis of an aircraft wing, which is critical for ensuring aircraft structural integrity. Of particular importance is the main wing spar, which is subject to aerodynamic lift loads at certain stations during flight. We consider the simulation set-up of a wing spar cantilever beam in the MATLAB module in \cite{mathworks_wingspar_rom}, which has 10 spanwise stations (see Figure \ref{fig:Metric comparison 11D} left). We investigate $d=11$ input variables: the first is dynamic pressure ($x_1$, with range [24500, 38281] Pa), and the rest are lift load multipliers at each station ($x_2, \cdots, x_{11}$, with range $[0.05,3.00]$). Given these inputs, the simulator makes use of a sophisticated structural dynamics model to simulate the maximum vertical deflection over the ten stations. As before, this simulator has stochastic noise.

Useful isotonicity information can again be elicited via a careful physical analysis. As dynamic pressure $x_1$ increases, the maximum vertical deflection should naturally increase. Similarly, as the lift load multipliers $x_2, \cdots, x_{11}$ increase at each station, this deflection should increase as well. We will inspect how well each model incorporates such isotonicity information in this higher-dimensional surrogate modeling problem. The same 30-minute time constraint on MCMC sampling is imposed here; to satisfy this, Proj-GP and Basis-GP are implemented with $N=2$, yielding $2^{11}=2048$ basis functions for the latter model. All models use the same $n=10d=110$-point LHD. The experiment is replicated 20 times.

Figure \ref{fig:Metric comparison 11D} (right) shows the test log-MSE and CRPS boxplots for this aircraft wing application. Similar to before, Proj-GP and Basis-GP perform worse than the standard GP, which again shows that such existing models may face computational challenges in high dimensions that outweigh the incorporation of isotonicity information. In contrast, our TAIM and TAAIM models perform considerably better than all benchmarks on both metrics, which shows the effectiveness of a flexible transformed additive framework for isotonic modeling with many inputs. Here, TAIM and TAAIM perform similarly.

\section{Conclusion}
\label{sec: conclusion}

This paper introduces a new Bayesian model, called TAIM, which tackles the problem of surrogate modeling with many input variables provided isotonicity information. Such information can often be elicited from physical systems, and its integration can enhance surrogate modeling with limited data. However, existing Bayesian isotonic models can face computational and statistical challenges in high dimensions, which may cause worse performance than standard models without isotonicity in practical applications. To address this, TAIM makes use of a flexible transformed additive isotonic modeling framework. Using a data-estimated link function and a monotone basis model with spike-and-slab basis weight priors, we show that TAIM enjoys statistical and computational advantages that can ``tame'' the curse-of-dimensionality. Prediction-wise, TAIM achieves a posterior contraction rate (up to log factors) of $\mathcal{O}(n^{-1/3})$ when the true function is in a transformed additive isotonic form under mild smoothness conditions. Computation-wise, posterior sampling can be performed via an efficient Gibbs sampler, with each sampling iteration requiring work linear in dimension $d$. We further present an extension of TAIM for modeling potential deviations from transformed additivity. Numerical experiments and two applications show the effectiveness of TAIM for isotonic surrogate modeling with many input variables.

\vspace{0.2cm}
\small
\spacingset{1.03}
\bibliographystyle{apalike}
\bibliography{reference}


\clearpage
\setcounter{section}{0}
\setcounter{equation}{0}
\setcounter{theorem}{0}
\renewcommand{\thesection}{S\arabic{section}}
\renewcommand{\theequation}{S\arabic{equation}}
\renewcommand{\thetheorem}{S\arabic{theorem}}
\renewcommand{\theHsection}{S\arabic{section}}
\renewcommand{\theHequation}{S\arabic{equation}}
\renewcommand{\theHtheorem}{S\arabic{theorem}}
\allowdisplaybreaks
\normalsize
\spacingset{1.55}

\begin{center}
{\LARGE\bf Supplementary Materials}
\end{center}

\section{Proofs of Theorems 1 and 2}\label{sec: proofs}

In what follows, we present the proofs for Theorems 1 and 2 of the main paper, which concern the posterior contraction rates in the 1-d and multi-dimensional settings, respectively. For simplicity, we prove this for the specific case of $\mu_0 = 0$ and $\tau_0^2 = \tau_g^2 = \tau_1^2 = \cdots = \tau_d^2 = 1$, since the case with general $\mu_0, \tau_g^2, \tau_1^2, \cdots, \tau_d^2$ follows immediately via rescaling and affects only constants in the contraction rate. For a function $f_0$ and radius $\epsilon > 0$, let
$B(f_0, \epsilon) = \{f: \|f - f_0\|_\infty < \epsilon\}$ denote the sup-norm ball around $f_0$, and let
$\cC(\eta, \cS, \|\cdot\|)$ denote the $\eta$-covering number of a set $\cS$ with respect to a norm $\|\cdot\|$, i.e., the minimal number of points $s_1, \cdots, s_m$ such that every $s \in \cS$ satisfies $\|s - s_j\| \le \eta$ for some $1 \le j \le m$ (see Appendix C of \cite{ghosal2017fundamentals} for details). We call such a set of points $\{s_1, \cdots, s_m\}$ the $\eta$-net of $\cS$. Further let $\Pi(\cdot)$ denote the prior distribution on $f$ (or $g$ in the 1-d setting). 

We first prove an intermediate lemma in Section \ref{ssec: small ball probability} that lower bounds the small ball probability, then use this to show the 1-d posterior contraction rate in Theorem 1 of the main paper. We then use these results to prove the posterior contraction rate for TAIM in Theorem 2 of the main paper.

\subsection{Small ball probability bound}\label{ssec: small ball probability}

Following Theorem 1 of the main paper, let $g$ follow the basis model in Equation \eqref{eqn: monotone link function} of the main paper (with domain constrained to $[0,1]$) with discretization level $N$ and priors following Equation \eqref{eqn: weightprior2} of the main paper. 
The following lemma shows that the basis model on $g$, equipped with spike-and-slab priors on basis weights, provides a lower bound on the ``small ball probability'' \cite{van2011information}, i.e., the prior probability of the sup-norm ball $B(g_0, \epsilon)$ around the true function $g_0$.
\begin{lemma}\label{lem: prior probability lower bound}
    For a monotone increasing $g_0 \in C^{\alpha}[0, 1]$, select $G_2>0$ to satisfy $G_2 \ge g_0(1) - g_0(0)$. Then there exists $C_1 >0$ such that, for all $\epsilon \in (0,1]$ and $N \ge C_1\epsilon^{-1/(\alpha \land 1)}$,
    \begin{equation}
        \Pi\big\{B(g_0, \epsilon)\big\}  \ge \exp\left\{-\frac{4G_2}{\epsilon} \log \frac{4N\sqrt{2\pi}}{\zeta\epsilon} - \frac{3G_2\epsilon}{4} - 2\zeta - g_0(0)^2 - \epsilon^2 + \log \frac{\epsilon}{2\sqrt{2\pi}}\right\}.
    \end{equation}
\end{lemma}

\begin{proof}
Since $g_0 \in C^{\alpha}[0, 1]$, there exists $C_1 > 0$ such that for all
$N \ge C_1 \epsilon^{-1/(\alpha \land 1)}$,
\begin{equation}\label{eqn: grid oscillation}
    |g_0(x) - g_0(y)| < \epsilon/4
\end{equation}
holds whenever $|x - y| \le 1/N$. We fix such an $N$ and construct an increasing integer
sequence $\{k_j\}_{0 \le j \le J}$ as:
\begin{equation}
    k_{j+1} = \min \Big\{k: g_0\Big( \frac{k}{N}\Big) - g_0\Big( \frac{k_j}{N}\Big) > \frac{\epsilon}{4}, \, 1 \le k \le N \Big\},
\end{equation}
with $k_0 = 0$ and the last index $J$ satisfying
$g_0(1) - g_0(k_J/N) \le \epsilon/4$. For $G_2 \ge g_0(1) - g_0(0)$, $J < 4G_2/\epsilon$ holds
since:
\begin{equation}\label{eqn: J bound}
    G_2 \ge \sum_{j=1}^J \Big\{g_0 \Big(\frac{k_j}{N}\Big) - g_0\Big(\frac{k_{j-1}}{N} \Big) \Big\} > \frac{J\epsilon}{4}.
\end{equation}
We define $A = \{k_1, \cdots, k_J\}$, $\alpha_0^* = g_0(0)$, and:
\begin{equation}\label{eqn: target coefficients}
    \alpha^*_{k_j} = g_0\Big(\frac{k_j}{N}\Big) - g_0\Big(\frac{k_{j-1}}{N}\Big)
\end{equation}
for $1 \le j \le J$. By the minimality of $k_j$ together with \eqref{eqn: grid oscillation},
these coefficients satisfy $\epsilon/4 < \alpha^*_{k_j} \le \epsilon/2$. Letting:
\begin{equation}
    g^*_N(x) = \alpha^*_0 + \sum_{j=1}^J \alpha^*_{k_j} \psi_{k_j}(x),
\end{equation}
we obtain $g^*_N(k_l/N) = g_0(k_l/N)$ for all $0 \le l \le J$, since
$\psi_{k_j}(k_l/N) = 1_{k_j \le k_l}$ and the resulting sum telescopes. As $g_0$ and $g^*_N$ are
both monotone increasing and agree at these points, their deviation on each block
$[k_j/N, k_{j+1}/N]$ is at most the increment of $g_0$ over that block, which is bounded by
$\epsilon/2$ from \eqref{eqn: target coefficients}. Thus,
$\|g^*_N - g_0\|_\infty \le \epsilon/2$, and applying the triangle inequality, we have:
\begin{equation}\label{eqn: reduction to hNstar}
    \Pi\big(B(g_0, \epsilon)\big) \ge \Pr\big(\|g - g^*_N\|_\infty < \epsilon/2 \big).
\end{equation}

Define the matrix $B \in \bbR^{J \times J}$ satisfying $B_{ij} = 1_{i \le j}$. For $v \in \bbR^J$, the function
$\sum_j v_j \psi_{k_j}$ is piecewise linear and vanishes at the origin, so its supremum norm is
attained at a knot and $\|\sum_j v_j \psi_{k_j}\|_\infty = \|B^{\!\top} v\|_\infty$. We consider
the event:
\begin{equation}\label{eqn: coefficient event}
    \mathcal{E} = \Big\{ I_m = 1_{m \in A} \ \forall m \Big\}
    \cap \Big\{ |\alpha_0 - \alpha_0^*| < \frac{\epsilon}{4} \Big\}
    \cap \Big\{ \big\|B^{\!\top}\big(Z_A - \alpha^*_A\big)\big\|_\infty < \frac{\epsilon}{8} \Big\},
\end{equation}
where $Z_A = (Z_{k_1}, \cdots, Z_{k_J})^{\!\top}$ and
$\alpha^*_A = (\alpha^*_{k_1}, \cdots, \alpha^*_{k_J})^{\!\top}$. On the event $\cE$, we have
$g - g^*_N = (\alpha_0 - \alpha_0^*) + \sum_j (Z_{k_j} - \alpha^*_{k_j})\psi_{k_j}$, and hence
$\|g - g^*_N\|_\infty < \epsilon/4 + \epsilon/8 < \epsilon/2$. From the mutual independence of
$\alpha_0$, $\{I_m\}$ and $\{Z_m\}$, we obtain:
\begin{equation}\label{eqn: three factor split}
    \Pr\big(\|g - g^*_N\|_\infty < \epsilon/2 \big) \ge \Pr(\mathcal{E})
    = p_0 \, p_1 \Big(\frac{\zeta}{N}\Big)^{J}\Big(1 - \frac{\zeta}{N}\Big)^{N-J},
\end{equation}
with $p_0 = \Pr(|\alpha_0 - \alpha_0^*| < \epsilon/4)$ and
$p_1 = \Pr(\|B^{\!\top}(Z_A - \alpha^*_A)\|_\infty < \epsilon/8)$.

Let $\phi(\cdot)$ and $\Phi(\cdot)$ denote the probability density and cumulative distribution functions
of the standard normal distribution, respectively. Then:
\begin{align}
    p_0 &= \Phi\big(\alpha_0^* + \epsilon/4\big) - \Phi\big(\alpha_0^* - \epsilon/4\big) \nonumber \\
    & \ge \frac{\epsilon}{2} \phi\big(|\alpha_0^*| + \epsilon/4\big) \nonumber \\
    &\ge \frac{\epsilon}{2\sqrt{2\pi}} \exp\Big(-\frac{\{|g_0(0)| + \epsilon/4\}^2}{2}\Big)
    \ge \frac{\epsilon}{2\sqrt{2\pi}} \exp\big(- g_0(0)^2 - \epsilon^2 \big).
    \label{eqn: p0 lower bound}
\end{align}
For $p_1$, we set $S = \{z \in \bbR^J : \|B^{\!\top}(z - \alpha_A^*)\|_\infty < \epsilon/8\}$, so
that $p_1 = \Pr(Z_A \in S)$. Since $(B^{\!\top} v)_l = \sum_{i \le l} v_i$, consecutive
coordinates of $B^{\!\top}(z - \alpha^*_A)$ differ by exactly $z_l - \alpha^*_{k_l}$, so any
$z \in S$ satisfies $|z_l - \alpha^*_{k_l}| < \epsilon/8 + \epsilon/8 = \epsilon/4$ for all $l$;
combined with $\alpha^*_{k_l} > \epsilon/4$ from \eqref{eqn: target coefficients}, this gives
$z_l > 0$ for all $l$ and hence $S \subset (0, \infty)^J$. The coordinates of $Z_A$ are
independent $\cN^+(0, 1)$ variables, whose density on $(0, \infty)$ is $2\phi$, so for
$\widetilde Z \sim \cN_J(0, I)$, we have:
\begin{equation}\label{eqn: half normal to normal}
    p_1 = \int_S \prod_{l=1}^J 2\phi(z_l)\, dz = 2^{J}\Pr(\widetilde Z \in S)
    \ge \Pr(\widetilde Z \in S).
\end{equation}
Here, the condition $S \subset (0, \infty)^J$ allows the half-normal density to be
replaced by $2\phi$ on all of $S$.

Since $\widetilde Z$ has density $(2\pi)^{-J/2}\exp(-\|z\|_2^2/2)$, we bound
$\Pr(\widetilde Z \in S)$ below by the volume of $S$ times the infimum of this density over $S$.
Since $B$ is unit upper triangular, $\det B = 1$, so $S$, the preimage of the cube
$(-\epsilon/8, \epsilon/8)^J$ under $z \mapsto B^{\!\top}(z - \alpha^*_A)$, has volume
$(\epsilon/4)^J$. Moreover, every $z \in S$ satisfies $|z_l - \alpha^*_{k_l}| < \epsilon/4$ as
shown above, so $\|z\|_2^2 \le 2\|\alpha^*_A\|_2^2 + 2\|z - \alpha^*_A\|_2^2
< 2\|\alpha^*_A\|_2^2 + J\epsilon^2/8$. Finally, the increments $\alpha^*_{k_j}$ telescope, so
$\sum_{j=1}^J \alpha^*_{k_j} = g_0(k_J/N) - g_0(0) \le G_2$, and together with
$\alpha^*_{k_j} \le \epsilon/2$ this gives $\|\alpha^*_A\|_2^2 \le \max_j \alpha_{k_j}^* \sum_{j=1}^J \alpha_{k_j}^* \le G_2\epsilon/2$. Combining with \eqref{eqn: half normal to normal}, we obtain:
\begin{equation}\label{eqn: p1 lower bound}
    p_1 \ge \Pr(\widetilde Z \in S) \ge \exp\Big(-\frac{G_2\epsilon}{2}\Big)
    \Big(\frac{\epsilon}{4\sqrt{2\pi}}\Big)^{J} \exp\Big(-\frac{J\epsilon^2}{16}\Big).
\end{equation}
Also, $(1 - \zeta/N)^{N-J} \ge (1 - \zeta/N)^{N} \ge \exp(-2\zeta)$ holds for $N > 2\zeta$,
since $\log(1 - x) \ge -2x$ for $0 < x \le 1/2$. Substituting these and
\eqref{eqn: p0 lower bound} into \eqref{eqn: three factor split} and taking logarithms, we have:
\begin{equation}\label{eqn: log assembly}
    \log \Pr(\cE) \ge
    \Big(\log \frac{\epsilon}{2\sqrt{2\pi}} - g_0(0)^2 - \epsilon^2\Big)
    + \Big(-\frac{G_2\epsilon}{2} + J \log \frac{\epsilon}{4\sqrt{2\pi}}
    - \frac{J\epsilon^2}{16}\Big)
    + \Big(J \log \frac{\zeta}{N} - 2\zeta\Big).
\end{equation}
Here, two terms proportional to $J$ combine into
$-J \log\{4N\sqrt{2\pi}/(\zeta\epsilon)\}$, which is negative since $N > 2\zeta$. Therefore, replacing $J$ with a larger quantity $4G_2/\epsilon$ from \eqref{eqn: J bound} gives a looser lower bound, and similarly, we have $-J\epsilon^2/16 \ge -G_2\epsilon/4$. Combining these with
\eqref{eqn: reduction to hNstar}, we obtain:
\begin{equation}\label{eqn: exact lower bound on small ball probability}
    \Pi\big(B(g_0, \epsilon)\big) \ge
    \exp\Big\{ -\frac{4G_2}{\epsilon}\log\frac{4N\sqrt{2\pi}}{\zeta\epsilon}
    - \frac{3G_2\epsilon}{4} - 2\zeta - g_0(0)^2 - \epsilon^2
    + \log\frac{\epsilon}{2\sqrt{2\pi}} \Big\}.
\end{equation}
Taking $C_1$ to be large enough such that both \eqref{eqn: grid oscillation} and $N > 2\zeta$ hold concludes
the proof.
\end{proof}

\subsection{Proof of Theorem \ref{thm: 1d}}\label{ssec: 1D rate proof}

We now present the proof of Theorem 1 in the main paper. This proof follows the general strategy for proving posterior contraction rates outlined in \cite{ghosal2000convergence, ghosal2007convergence}. We first prove several lemmas, which together with the earlier Lemma \ref{lem: prior probability lower bound}, allow us to show the posterior contraction rate claimed in Theorem 1 of the main paper.

For fixed-design regression with Gaussian noise, the following lemma provides a sufficient condition on the function prior for achieving a posterior contraction rate of $\epsilon_n$. This is a variant of Theorem 8.26 of \citet{ghosal2017fundamentals} (see also Section 7.7 of \citet{ghosal2007convergence}), and we give a direct proof similar to Theorem 2.1 of \citet{ghosal2000convergence}.

\begin{lemma}\label{lem: gaussian contraction}
    Let $\mathbf{x}_1, \cdots, \mathbf{x}_n$ be fixed points on the domain $\mathcal{X}$, and let
    $y_i = f_0(\mathbf{x}_i) + \epsilon_i$ with $\epsilon_i \overset{i.i.d.}{\sim} \cN(0, \sigma_\epsilon^2)$
    for a known $\sigma_\epsilon^2$ and a given function $f_0 : \mathcal{X} \to \bbR$. Let $\Pi$ be the considered
    prior distribution on the functions $f : \mathcal{X} \to \bbR$, and let $\epsilon_n \to 0$ with
    $n\epsilon_n^2 \to \infty$. Suppose there exist sets $\cF_n$ of functions $f$ such that, for sufficiently large $n$ and for constants $C_3, A' > 0$ and
    $D > C_3 + 2\sigma_\epsilon^{-2}$, we have:
    \begin{align}
        &\Pi\big(\|f - f_0\|_n < \epsilon_n\big) \ge \exp(-C_3 n\epsilon_n^2),
        \label{eqn: lemma prior mass}\\
        &\Pi(\cF_n^c) \le \exp(-D n\epsilon_n^2),
        \label{eqn: lemma sieve}\\
        &\log \cC\big(\epsilon_n/8, \cF_n, \|\cdot\|_n\big) \le A' n\epsilon_n^2.
        \label{eqn: lemma entropy}
    \end{align}
    These are typically referred to as the Kullback-Leibler (KL) ball, sieve and metric entropy conditions, respectively. Then there
    exists a $K > 0$ such that $\bbE_{f_0}\Pi(\|f - f_0\|_n > K\epsilon_n \mid \Data) \to 0$.
\end{lemma}

\begin{proof}
Write $\sigma = \sigma_\epsilon$. Replacing $y_i$, $f$ and $f_0$ by $y_i/\sigma$, $f/\sigma$ and
$f_0/\sigma$ changes $\epsilon_n$ to $\epsilon_n/\sigma$, and conditions
\eqref{eqn: lemma prior mass}--\eqref{eqn: lemma entropy} hold with $C_3, D, A'$
replaced by $\sigma^2 C_3, \sigma^2 D, \sigma^2 A'$. As $D > C_3 + 2\sigma^{-2}$ is equivalent to
$\sigma^2 D > \sigma^2 C_3 + 2$, we may thus assume that $\sigma = 1$ and $D > C_3 + 2$.

Let $p_f$ denote the density of $y = (y_1, \cdots, y_n)$ when $y_i = f(\mathbf{x}_i) + \epsilon_i$. {Since $\log(p_{f_0}/p_f)(y) = \sum_{i=1}^n \{f_0(\mathbf{x}_i) - f(\mathbf{x}_i)\}(y_i - f_0(\mathbf{x}_i)) + n\|f - f_0\|_n^2/2$, the Kullback--Leibler divergence $\mathrm{KL}(p_{f_0}; p_f) = \bbE_{f_0}\log(p_{f_0}/p_f)(y)$ and the centered second moment $V_{2,0}(p_{f_0}; p_f) = \bbE_{f_0}\{\log(p_{f_0}/p_f)(y) - \mathrm{KL}(p_{f_0}; p_f)\}^2$ satisfy $\mathrm{KL}(p_{f_0}; p_f) = n\|f - f_0\|_n^2/2$ and $V_{2,0}(p_{f_0}; p_f) = n\|f - f_0\|_n^2$. Therefore, the Kullback--Leibler neighborhood $B_{n,2}(f_0, \epsilon) = \{f : \mathrm{KL}(p_{f_0}; p_f) \le n\epsilon^2,\ V_{2,0}(p_{f_0}; p_f) \le n\epsilon^2\}$ contains the set $\{f : \|f - f_0\|_n < \epsilon\}$ for every $\epsilon > 0$.}  

With $R_f = (p_f/p_{f_0})(y)$, we have:
\begin{equation}\label{eqn: posterior and change of measure}
    \Pi(A \mid \Data) = \frac{\int_A R_f\, d\Pi(f)}{\int R_f\, d\Pi(f)}, \qquad
    \bbE_{f_0}\{g(y)R_f\} = \bbE_f\, g(y),
\end{equation}
for any measurable set $A$ and measurable $0 \le g \le 1$. Define the event:
\begin{equation}\label{eqn: evidence event}
    \cE_n = \Big\{\int R_f\, d\Pi(f) \ge e^{-(C_3 + 2)n\epsilon_n^2}\Big\}.
\end{equation}
Since $\Pi(B_{n,2}(f_0, \epsilon_n)) \ge e^{-C_3 n\epsilon_n^2}$ by \eqref{eqn: lemma prior mass},
Lemma 8.21 of \citet{ghosal2017fundamentals} with $k = 2$ gives
$P_{f_0}(\cE_n^c) \le (n\epsilon_n^2)^{-1}$ for sufficiently large $n$.

Let $K \ge 1$ and define the set $S_n = \{f \in \cF_n : \|f - f_0\|_n > K\epsilon_n\}$. We bound the posterior probability of $S_n$ through tests. For $f_1$ with $\|f_1 - f_0\|_n > 0$, Lemma 8.27(i) of \citet{ghosal2017fundamentals} gives
that:
\begin{equation}
    \phi_{f_1} = 1\Big\{\sum_{i=1}^n (f_1 - f_0)(\mathbf{x}_i)\{y_i - f_0(\mathbf{x}_i)\}
    > \frac{n\|f_1 - f_0\|_n^2}{4}\Big\},
\end{equation}
with the following satisfied:
\begin{equation}\label{eqn: test errors}
    \bbE_{f_0}\phi_{f_1} \le e^{-n\|f_1 - f_0\|_n^2/32}, \qquad
    \sup_{f : \|f - f_1\|_n \le \|f_1 - f_0\|_n/2} \bbE_f(1 - \phi_{f_1})
    \le e^{-n\|f_1 - f_0\|_n^2/32}.
\end{equation}
By \eqref{eqn: lemma entropy}, $\cF_n$ is covered by $N_n \le e^{A' n\epsilon_n^2}$
$\|\cdot\|_n$-balls of radius $\epsilon_n/8$. Choosing $f_l \in S_n$ in each ball that
intersects $S_n$, every $f \in S_n$ lies in a ball containing some $f_l$, so that:
\begin{equation}\label{eqn: net in Sn}
    \|f - f_l\|_n \le \frac{\epsilon_n}{8} + \frac{\epsilon_n}{8} = \frac{\epsilon_n}{4}
    < \frac{K\epsilon_n}{2} < \frac{\|f_l - f_0\|_n}{2}.
\end{equation}
Let $\phi_n = \max_{l \le N_n} \phi_{f_l}$. By \eqref{eqn: test errors}, a union bound, and
$1 - \phi_n \le 1 - \phi_{f_l}$ together with \eqref{eqn: net in Sn}, we have:
\begin{equation}
    \bbE_{f_0}\phi_n \le N_n e^{-K^2 n\epsilon_n^2/32} \le e^{(A' - K^2/32)n\epsilon_n^2},
    \qquad
    \sup_{f \in S_n}\bbE_f(1 - \phi_n) \le e^{-K^2 n\epsilon_n^2/32}.
\end{equation}

Finally, as $\{f : \|f - f_0\|_n > K\epsilon_n\} \subset \cF_n^c \cup S_n$,
\eqref{eqn: posterior and change of measure} and \eqref{eqn: evidence event} gives:
\begin{align}
    \Pi\big(\|f - f_0\|_n > K\epsilon_n \mid \Data\big)
    &\le 1_{\cE_n^c} + 1_{\cE_n}\big\{\phi_n + (1 - \phi_n)\,\Pi(\cF_n^c \cup S_n \mid \Data)\big\}
    \nonumber \\
    &\le 1_{\cE_n^c} + \phi_n + e^{(C_3 + 2)n\epsilon_n^2}(1 - \phi_n)
    \int_{\cF_n^c \cup S_n} R_f\, d\Pi(f).\label{eqn: posterior contraction complement decomposition bound}
\end{align}
By \eqref{eqn: posterior and change of measure} with $g = 1 - \phi_n$, we have:
\begin{equation}
    \bbE_{f_0}\Big\{(1 - \phi_n)\int_{\cF_n^c \cup S_n} R_f\, d\Pi(f)\Big\}
    = \int_{\cF_n^c \cup S_n} \bbE_f(1 - \phi_n)\, d\Pi(f)
    \le \Pi(\cF_n^c) + \sup_{f \in S_n}\bbE_f(1 - \phi_n).
\end{equation}
Taking expectations in \eqref{eqn: posterior contraction complement decomposition bound} and combining the bounds on
$P_{f_0}(\cE_n^c)$ and $\bbE_{f_0}\phi_n$, we get:
\begin{align}
    \bbE_{f_0}\Pi\big(\|f - f_0\|_n > K\epsilon_n \mid \Data\big)
    &\le \frac{1}{n\epsilon_n^2} + e^{(A' - K^2/32)n\epsilon_n^2}
    + e^{(C_3 + 2)n\epsilon_n^2}\Big\{\Pi(\cF_n^c) + \sup_{f \in S_n}\bbE_f(1 - \phi_n)\Big\}
    \nonumber \\
    &\le \frac{1}{n\epsilon_n^2} + e^{(A' - K^2/32)n\epsilon_n^2}
    + e^{-(D - C_3 - 2)n\epsilon_n^2} + e^{(C_3 + 2 - K^2/32)n\epsilon_n^2},
\end{align}
by \eqref{eqn: lemma sieve}. Since $D > C_3 + 2$ and $n\epsilon_n^2 \to \infty$, the right-hand side
tends to zero for $K^2 > 32\max\{A', C_3 + 2\}$, which concludes the proof.
\end{proof}

Next, for $k \in \bbN$ and
$M > 0$, let $\cG_{k, M}$ denote the set of functions $g$ in the form of Equation \eqref{eqn: monotone link function} of the main paper, whose coefficients satisfy:
\begin{equation}\label{eqn: 1d sieve}
    |\alpha_0| \le M, \qquad \max_{1 \le m \le N} \alpha_m \le M, \qquad
    \#\{m : \alpha_m \neq 0\} \le k .
\end{equation}
The following lemma provides an upper bound on the metric entropy of $\cG_{k, M}$:
 
\begin{lemma}\label{lem: 1d entropy}
    For $0 < \eta \le (k+1)M$,
    \begin{equation}
        \log \cC\big(\eta, \cG_{k, M}, \|\cdot\|_\infty\big)
        \le k \log N + (k+1) \log \frac{3(k+1)M}{\eta}.
    \end{equation}
\end{lemma}

\begin{proof}
    For every element of $\cG_{k, M}$, the set $\{m \ge 1 : \alpha_m \neq 0\}$ is contained in
    some $\cS \subset \{1, \cdots, N\}$ with $|\cS| = k \land N$, and there are
    $\binom{N}{k \land N} \le N^k$ such sets. Fix one such $\cS$ and let
    $\theta, \theta' \in [-M, M]^{p}$ be two coefficient vectors consisting of $\alpha_0$ and
    $\{\alpha_m\}_{m \in \cS}$, where $p = 1 + |\cS| \le k + 1$, with $p < k + 1$ only when $k > N$.
    We write $g_\theta, g_{\theta'}$ for the associated functions. Since $0 \le \psi_m \le 1$ and
    $|\cS| \le k$,
    \begin{equation}
        \|g_\theta - g_{\theta'}\|_\infty \le (k+1) \|\theta - \theta'\|_\infty .
    \end{equation}
    Covering $[-M, M]^{p}$ in $\|\cdot\|_\infty$ at resolution $\eta/(k+1)$ therefore yields
    an $\eta$-net {in $\|\cdot\|_\infty$ of $\{g_\theta : \theta \in [-M, M]^p\}$ for that $\cS$}. This requires at most
    $\{2M(k+1)/\eta + 1\}^{p} \le \{3(k+1)M/\eta\}^{k+1}$ points, where the last inequality
    uses $p \le k + 1$ and $\eta \le (k+1)M$, so that the base is at least one. Multiplying by the
    number of sets $\cS$ concludes the proof.
\end{proof}

With this, the following lemma then provides the sieve condition on $\cG_{k, M}$:
 
\begin{lemma}\label{lem: 1d sieve prior mass}
    For the prior \eqref{eqn: weightprior2} of the main paper with inclusion probability $\zeta/N$,
    \begin{equation}
        \Pi\big(\cG_{k, M}^c\big)
        \le \Big(\frac{e\zeta}{k}\Big)^{k} + (N + 2)\exp\Big(-\frac{M^2}{2}\Big).
    \end{equation}
\end{lemma}
 
\begin{proof}
    Let $T = \sum_{m=1}^{N} I_m \sim \mathrm{Bin}(N, \zeta/N)$. Since $\alpha_m = I_m Z_m$ with
    $Z_m \ge 0$, we have $\#\{m : \alpha_m \neq 0\} \le T$ and $\alpha_m \le Z_m$, so by
    \eqref{eqn: 1d sieve}:
    \begin{equation}
        \cG_{k, M}^c \subset \{T > k\} \cup \bigcup_{m=1}^N \{Z_m > M\} \cup \{|\alpha_0| > M\},
    \end{equation}
    and it suffices to bound the probability of each event. On $\{T > k\}$, all indicators in some
    subset of size $k$ equal one, so a union bound over the $\binom{N}{k}$ such subsets, together
    with $\binom{N}{k} \le N^k/k!$ and $k! \ge (k/e)^k$, gives:
    \begin{equation}
        \Pr(T > k) \le \binom{N}{k}\Big(\frac{\zeta}{N}\Big)^{k}
        \le \Big(\frac{eN}{k}\Big)^{k}\Big(\frac{\zeta}{N}\Big)^{k} = \Big(\frac{e\zeta}{k}\Big)^{k}.
    \end{equation}
    For the remaining events, writing $1 - \Phi(x) = \int_x^\infty \phi(t)\, dt$ for $x \ge 0$ and substituting $t = x + s$, we have:
    \begin{equation}\label{eqn: gaussian tail}
        1 - \Phi(x) = \int_0^\infty \phi(x + s)\, ds
        = e^{-x^2/2} \int_0^\infty \phi(s) e^{-xs}\, ds \le \frac{1}{2} e^{-x^2/2}.
    \end{equation}
    Hence $\Pr(Z_m > M) = 2\{1 - \Phi(M)\} \le \exp(-M^2/2)$ for $Z_m \sim \cN^+(0, 1)$, which
    sums to $N\exp(-M^2/2)$ over $m$, and $\Pr(|\alpha_0| > M) = 2\{1 - \Phi(M)\} \le 2\exp(-M^2/2)$.
    A union bound over the three events concludes the proof.
\end{proof}

We can now prove the posterior contraction rate in Theorem \ref{thm: 1d} of the main paper.

\begin{proof}[Proof of Theorem \ref{thm: 1d}]
Let $\epsilon_n = n^{-1/3}(\log n)^{1/3}$, let $C_1'$ be the constant of Lemma
\ref{lem: prior probability lower bound} for $g_0$, and let:
\begin{equation}\label{eqn: 1d rescaled rate}
    \lambda = \max\big\{1, (C_1'/C_1)^{\alpha \land 1}\big\}, \qquad
    \tilde\epsilon_n = \lambda\epsilon_n, \qquad
    n\tilde\epsilon_n^2 = \lambda^2 n^{1/3}(\log n)^{2/3}.
\end{equation}
Write $B_n(g_0, \epsilon) = \{g : \|g - g_0\|_n < \epsilon\}$. As $\lambda$ does not depend on $n$,
by Lemma \ref{lem: gaussian contraction} with $\tilde\epsilon_n$ in place of $\epsilon_n$,
$\mathcal{X} = [0, 1]$ and $f_0 = g_0$, it suffices to find sets $\cG_n$ and constants $C_3$,
$D > C_3 + 2\sigma_\epsilon^{-2}$ and $A'$ such that, for sufficiently large $n$:
\begin{equation}\label{eqn: 1d three conditions}
    \Pi\big(B_n(g_0, \tilde\epsilon_n)\big) \ge e^{-C_3 n\tilde\epsilon_n^2}, \qquad
    \Pi(\cG_n^c) \le e^{-Dn\tilde\epsilon_n^2}, \qquad
    \log \cC\big(\tilde\epsilon_n/8, \cG_n, \|\cdot\|_n\big) \le A' n\tilde\epsilon_n^2.
\end{equation}
Then \eqref{eqn: 1drate} of the main paper holds with the constant $K\lambda$, for $K$ given by Lemma
\ref{lem: gaussian contraction}.

For the first condition, let $G_2 = \max\{g_0(1) - g_0(0), 1\}$. The resolution condition of
Theorem \ref{thm: 1d} gives:
\begin{equation}
    N \ge C_1\epsilon_n^{-1/(\alpha \land 1)}
    = C_1\lambda^{1/(\alpha \land 1)}\tilde\epsilon_n^{-1/(\alpha \land 1)}
    \ge C_1'\tilde\epsilon_n^{-1/(\alpha \land 1)},
\end{equation}
so Lemma \ref{lem: prior probability lower bound} applies with $\epsilon = \tilde\epsilon_n$.
Hence, for sufficiently large $n$:
\begin{align}
    \log \Pi\big(B_n(g_0, \tilde\epsilon_n)\big)
    &\ge \log \Pi\big(B(g_0, \tilde\epsilon_n)\big)
    \ge -\frac{4G_2}{\tilde\epsilon_n}\log \frac{4N\sqrt{2\pi}}{\zeta\tilde\epsilon_n} - C'
    - \log\frac{1}{\tilde\epsilon_n}
    \nonumber \\
    &\ge -\frac{C_2}{\tilde\epsilon_n}\log \frac{N}{\tilde\epsilon_n}
    \ge -C_2(c_0 + 1)\frac{\log n}{\tilde\epsilon_n}
    \ge -C_3 n\tilde\epsilon_n^2,
    \label{eqn: 1d prior mass condition}
\end{align}
where $C', C_2$ are constants depending only on $G_2$, $\zeta$ and $g_0(0)$, and
$C_3 = C_2(c_0 + 1)$. The first inequality uses $\|\cdot\|_n \le \|\cdot\|_\infty$, the second is
\eqref{eqn: exact lower bound on small ball probability}, and the last two use
$\log(N/\tilde\epsilon_n) \le (c_0 + 1)\log n$, which follows from $\log N \le c_0 \log n$ for some
$c_0 > 0$, and $\log n/\tilde\epsilon_n = \lambda^{-3} n\tilde\epsilon_n^2$.

For the second condition, fix $D > C_3 + 2\sigma_\epsilon^{-2}$, and let:
\begin{equation}
    k_n = \big\lceil A n^{1/3}(\log n)^{-1/3}\big\rceil, \qquad M_n = n^{1/3},
    \qquad \cG_n = \cG_{k_n, M_n},
\end{equation}
with $A > 4\lambda^2 D$. Since $\log(k_n/e\zeta) \ge (\log n)/4$ for sufficiently large $n$, Lemma
\ref{lem: 1d sieve prior mass} gives:
\begin{equation}
    \Pi(\cG_n^c) \le \exp\Big(-k_n \log \frac{k_n}{e\zeta}\Big) + (N + 2)e^{-n^{2/3}/2}
    \le e^{-A n^{1/3}(\log n)^{2/3}/4} + (N + 2)e^{-n^{2/3}/2}
    \le e^{-Dn\tilde\epsilon_n^2}
\end{equation}
for sufficiently large $n$, where the last inequality uses
$A n^{1/3}(\log n)^{2/3}/4 = A n\tilde\epsilon_n^2/(4\lambda^2) > Dn\tilde\epsilon_n^2$,
$\log N \le c_0 \log n$ and $n^{2/3}/(n\tilde\epsilon_n^2) \to \infty$.

For the third condition, Lemma \ref{lem: 1d entropy} with
$\eta = \tilde\epsilon_n/8 \le (k_n + 1)M_n$ gives, for sufficiently large $n$:
\begin{equation}\label{eqn: 1d entropy condition}
    \log \cC\big(\tilde\epsilon_n/8, \cG_n, \|\cdot\|_n\big)
    \le k_n \log N + (k_n + 1)\log\frac{24(k_n + 1)M_n}{\tilde\epsilon_n}
    \le (k_n + 1)(c_0 + 2)\log n
    \le A' n\tilde\epsilon_n^2,
\end{equation}
where $A' = 2A(c_0 + 2)$. The first inequality uses $\|\cdot\|_n \le \|\cdot\|_\infty$, the
second uses $\log N \le c_0 \log n$ and $24(k_n + 1)M_n/\tilde\epsilon_n \le n^2$, and the last
uses $k_n + 1 \le 2A n^{1/3}(\log n)^{-1/3}$ and $\lambda \ge 1$. This verifies
\eqref{eqn: 1d three conditions}, and concludes the proof.
\end{proof}

\subsection{Proof of Theorem \ref{thm: TrAM posterior contraction rate}}\label{ssec: posterior contraction general d}

We now present the proof of Theorem \ref{thm: TrAM posterior contraction rate} in the main paper, on the posterior contraction rate of the TAIM model. The proof strategy is similar to that for Theorem \ref{thm: 1d}. We first construct a sieve adapted to the TAIM prior, similar to the sieve defined in Section \ref{ssec: 1D rate proof}. For
$k \in \bbN$ and $M > 0$, let $\cF_{k, M}$ denote the set of functions
$f = g(\sum_{l=1}^d h_l)$ of the form
\eqref{eqn: generalized additive model}, \eqref{eqn: uniadd} and \eqref{eqn: monotone link function} of the main paper whose coefficients
satisfy:
\begin{equation}\label{eqn: sieve}
    |\alpha_0| \le M, \quad
    \max_{j} \alpha_j \vee \max_{l, j} \beta^{(l)}_j \le M, \quad
    \#\{j : \alpha_j \neq 0\} \le k, \quad
    \max_{1 \le l \le d} \#\{j : \beta^{(l)}_j \neq 0\} \le k,
\end{equation}
and write $\bar N = \max\{N_g, N_1, \cdots, N_d\}$. The following two lemmas control the
entropy of $\cF_{k, M}$ and the prior mass of its complement.

\begin{lemma}\label{lem: entropy}
    For $\bar N M \ge 1$, $k \ge 2$ and $0 < \eta \le (d+1)kM$,
    \begin{equation}
        \log \cC\big(\eta, \cF_{k, M}, \|\cdot\|_\infty\big)
        \le (d+1)k \log \bar N
        + \big\{(d+1)k + 1\big\} \log \frac{6d \bar N k^2 M^2}{\eta}.
    \end{equation}
\end{lemma}

\begin{proof}
    For $f \in \cF_{k, M}$, the indices of the nonzero coefficients among $\{\alpha_j\}_{j=1}^{N_g}$
    are contained in some $\cS_g \subset \{1, \cdots, N_g\}$ with $|\cS_g| = k \land N_g$, and for
    each $l$, those among $\{\beta^{(l)}_j\}_{j=1}^{N_l}$ are contained in some
    $\cS_l \subset \{1, \cdots, N_l\}$ with $|\cS_l| = k \land N_l$. As in the proof of Lemma
    \ref{lem: 1d entropy}, we construct a net separately for each choice of
    $(\cS_g, \cS_1, \cdots, \cS_d)$, and the number of such choices is at most
    $\binom{N_g}{k \land N_g}\prod_{l=1}^d \binom{N_l}{k \land N_l} \le \bar N^{(d+1)k}$. Fix one
    such choice, and let $\theta, \theta' \in [-M, M]^{p}$ be two coefficient vectors
    consisting of $\alpha_0$, $\{\alpha_j\}_{j \in \cS_g}$ and $\{\beta^{(l)}_j\}_{j \in \cS_l}$,
    $1 \le l \le d$, where $p = 1 + |\cS_g| + \sum_{l=1}^d |\cS_l| \le (d+1)k + 1$. Here $\alpha_0$ is
    the only intercept, as $h_l$ in \eqref{eqn: uniadd} of the main paper has none, and $p < (d+1)k + 1$ only when
    $k$ exceeds $N_g$ or some $N_l$. We write $f_\theta = g_\theta(H_\theta)$ and
    $f_{\theta'} = g_{\theta'}(H_{\theta'})$ for the associated functions, where
    $H_\theta = \sum_{l=1}^d h_{l, \theta}$. Since
    $0 \le \psi^{(l)}_j \le 1$ and each component has at most $k$ coefficients besides $\alpha_0$:
    \begin{equation}
        \|H_\theta - H_{\theta'}\|_\infty \le dk \|\theta - \theta'\|_\infty,
        \qquad
        \|g_\theta - g_{\theta'}\|_\infty \le (k+1) \|\theta - \theta'\|_\infty .
    \end{equation}
    For functions $g_a, g_b : \bbR \to \bbR$ and $H_a, H_b : [0, 1]^d \to \bbR$ such that $g_b$ is
    Lipschitz continuous with constant $L_b$, inserting $g_b(H_a(\mathbf{x}))$ gives, for every
    $\mathbf{x} \in [0, 1]^d$:
    \begin{align}
        \big|g_a(H_a(\mathbf{x})) - g_b(H_b(\mathbf{x}))\big|
        &\le \big|g_a(H_a(\mathbf{x})) - g_b(H_a(\mathbf{x}))\big|
        + \big|g_b(H_a(\mathbf{x})) - g_b(H_b(\mathbf{x}))\big| \nonumber \\
        &\le \|g_a - g_b\|_\infty + L_b \|H_a - H_b\|_\infty .
        \label{eqn: composition bound}
    \end{align}
    Each $\psi^{(g)}_j$ has slope at most $N_g$, so $g_{\theta'}$ is Lipschitz continuous with
    constant $N_g \sum_{j \in \cS_g} |\alpha'_j| \le \bar N k M$, where $\alpha'_j$ are the
    corresponding entries of $\theta'$. Applying \eqref{eqn: composition bound} with
    $(g_a, H_a) = (g_\theta, H_\theta)$ and $(g_b, H_b) = (g_{\theta'}, H_{\theta'})$, we have:
    \begin{equation}
        \|f_\theta - f_{\theta'}\|_\infty
        \le \|g_\theta - g_{\theta'}\|_\infty + \bar N k M \|H_\theta - H_{\theta'}\|_\infty
        \le 2d\bar N k^2 M \|\theta - \theta'\|_\infty .
    \end{equation}
    Covering $[-M, M]^{p}$ in $\|\cdot\|_\infty$ at resolution
    $\eta/(2d\bar N k^2 M)$ thus yields an $\eta$-net for that choice, and requires at most
    $\{4d\bar N k^2 M^2/\eta + 1\}^{p} \le (6d\bar N k^2 M^2/\eta)^{(d+1)k+1}$ points, where
    the last inequality uses $p \le (d+1)k + 1$ and $\eta \le (d+1)kM \le 2d\bar N k^2 M^2$, so that
    the base is at least one. Multiplying by the number of choices of
    $(\cS_g, \cS_1, \cdots, \cS_d)$ concludes the proof.
\end{proof}

Subsequently, the following Lemma provides the sieve condition on $\cF_{k, M}$.
\begin{lemma}\label{lem: sieve prior mass}
    Let $\zeta_{\max} = \max\{\zeta_g, \zeta_1, \cdots, \zeta_d\}$. Then:
    \begin{equation}
        \Pi\big(\cF_{k, M}^c\big)
        \le (d+1)\Big(\frac{e\zeta_{\max}}{k}\Big)^{k}
        + \big\{(d+1)\bar N + 2\big\}\exp\Big(-\frac{M^2}{2}\Big).
    \end{equation}
\end{lemma}

\begin{proof}
    Let $T_g = \sum_{j=1}^{N_g} I^{(g)}_j \sim \mathrm{Bin}(N_g, \zeta_g/N_g)$ and
    $T_l = \sum_{j=1}^{N_l} I^{(l)}_j \sim \mathrm{Bin}(N_l, \zeta_l/N_l)$ for $1 \le l \le d$.
    Since $\alpha_j = I^{(g)}_j Z^{(g)}_j$ and $\beta^{(l)}_j = I^{(l)}_j Z^{(l)}_j$ with
    $Z^{(g)}_j, Z^{(l)}_j \ge 0$, we have $\#\{j : \alpha_j \neq 0\} \le T_g$,
    $\#\{j : \beta^{(l)}_j \neq 0\} \le T_l$, $\alpha_j \le Z^{(g)}_j$ and
    $\beta^{(l)}_j \le Z^{(l)}_j$, so by \eqref{eqn: sieve}:
    \begin{equation}
        \cF_{k, M}^c \subset \{T_g > k\} \cup \bigcup_{l=1}^d \{T_l > k\}
        \cup \bigcup_{j=1}^{N_g} \{Z^{(g)}_j > M\}
        \cup \bigcup_{l=1}^d \bigcup_{j=1}^{N_l} \{Z^{(l)}_j > M\}
        \cup \{|\alpha_0| > M\}.
    \end{equation}
    The binomial bound in the proof of Lemma \ref{lem: 1d sieve prior mass} gives
    $\Pr(T_g > k) \le (e\zeta_g/k)^k$ and $\Pr(T_l > k) \le (e\zeta_l/k)^k$, each at most
    $(e\zeta_{\max}/k)^k$. By \eqref{eqn: gaussian tail}, each of the
    $N_g + \sum_l N_l \le (d+1)\bar N$ magnitude variables exceeds $M$ with probability at most
    $\exp(-M^2/2)$, and $\Pr(|\alpha_0| > M) \le 2\exp(-M^2/2)$. A union bound over these events
    concludes the proof.
\end{proof}

We can now use these results to prove the posterior contraction rate in Theorem \ref{thm: TrAM posterior contraction rate} of the main paper.

\begin{proof}[Proof of Theorem \ref{thm: TrAM posterior contraction rate}]
Without loss of generality, $h_{0,l}(0) = 0$ for each $l$ and $\sum_{l=1}^d h_{0,l}$ takes values
in $[0, 1]$; otherwise we replace $h_{0,l}$ by $\{h_{0,l} - h_{0,l}(0)\}/\sum_{l'} B_{l'}$ with
$B_l = h_{0,l}(1) - h_{0,l}(0)$, or by $h_{0,l} \equiv 0$ if $\sum_{l'} B_{l'} = 0$, and rescale or shift the argument of $g_0$ accordingly, which preserves
the monotonicity and the H\"older continuity of each $h_{0,l}$ and the Lipschitz continuity of
$g_0$. Let $L$ be the Lipschitz constant of $g_0$, and extend $g_0$ to $\bbR$ by
$g_0(z) = g_0(0)$ for $z < 0$ and $g_0(z) = g_0(1)$ for $z > 1$, which keeps it $L$-Lipschitz.
Since each $\psi^{(g)}_j$ in \eqref{eqn: psi definition} of the main paper is constant outside $[0, 1]$, so is $g$,
and hence $\sup_{z \in \bbR}|g(z) - g_0(z)| = \|g - g_0\|_\infty$.

Let $\epsilon_n = n^{-1/3}(\log n)^{1/3}$, let $C_g'$ and $C_l'$ be the constants of Lemma
\ref{lem: prior probability lower bound} for $g_0$ with $\alpha = 1$ and for $h_{0,l}$,
respectively, and let:
\begin{equation}\label{eqn: rescaled rate}
    \lambda = \max\Big\{1, \frac{2C_g'}{C_g},
    2Ld\Big(\frac{C_l'}{C_l}\Big)^{\alpha_l \land 1} \ (1 \le l \le d)\Big\}, \qquad
    \tilde\epsilon_n = \lambda\epsilon_n, \qquad
    n\tilde\epsilon_n^2 = \lambda^2 n^{1/3}(\log n)^{2/3}.
\end{equation}
Since $\epsilon_n^{-1} = (n/\log n)^{1/3}$, the resolution conditions of Theorem
\ref{thm: TrAM posterior contraction rate} give:
\begin{equation}\label{eqn: resolution rescaled}
    N_g \ge C_g\epsilon_n^{-1} \ge C_g'\Big(\frac{\tilde\epsilon_n}{2}\Big)^{-1}, \qquad
    N_l \ge C_l\epsilon_n^{-1/(\alpha_l \land 1)}
    \ge C_l'\Big(\frac{\tilde\epsilon_n}{2Ld}\Big)^{-1/(\alpha_l \land 1)}.
\end{equation}
Write $B_n(f_0, \epsilon) = \{f : \|f - f_0\|_n < \epsilon\}$. As $\lambda$ does not depend on $n$,
by Lemma \ref{lem: gaussian contraction} with $\tilde\epsilon_n$ in place of $\epsilon_n$ and
$\mathcal{X} = [0, 1]^d$, it suffices to find sets $\cF_n$ and constants $C_3$,
$D > C_3 + 2\sigma_\epsilon^{-2}$ and $A'$ such that, for all large $n$,
\begin{equation}\label{eqn: three conditions}
    \Pi\big(B_n(f_0, \tilde\epsilon_n)\big) \ge e^{-C_3 n\tilde\epsilon_n^2}, \qquad
    \Pi(\cF_n^c) \le e^{-Dn\tilde\epsilon_n^2}, \qquad
    \log \cC\big(\tilde\epsilon_n/8, \cF_n, \|\cdot\|_n\big) \le A' n\tilde\epsilon_n^2;
\end{equation}
then the claim holds with the constant $K\lambda$, for $K$ given by Lemma
\ref{lem: gaussian contraction}.

For the first condition, \eqref{eqn: composition bound} with $(g_a, H_a) = (g, \sum_l h_l)$ and
$(g_b, H_b) = (g_0, \sum_l h_{0,l})$ gives:
\begin{equation}\label{eqn: small ball decomposition}
    \|f - f_0\|_\infty \le \|g - g_0\|_\infty + L \sum_{l=1}^d \|h_l - h_{0,l}\|_\infty.
\end{equation}
Since the priors on $\{\alpha_j\}$ and $\{\beta^{(l)}_j\}_{l=1}^d$ are independent, we have:
\begin{equation}\label{eqn: product of small balls}
    \Pi\big(B(f_0, \epsilon)\big) \ge
    \Pi\Big(B\big(g_0, \tfrac{\epsilon}{2}\big)\Big)
    \prod_{l=1}^d \Pi\Big(B\big(h_{0,l}, \tfrac{\epsilon}{2Ld}\big)\Big),
\end{equation}
for every $\epsilon > 0$. By \eqref{eqn: resolution rescaled}, Lemma
\ref{lem: prior probability lower bound} applies to $g_0$ with $\alpha = 1$, $G_2 = L$ and
radius $\tilde\epsilon_n/2$, and to each $h_{0,l}$ with $G_2 = 1$ and radius
$\tilde\epsilon_n/(2Ld)$, where the factor $p_0$ of \eqref{eqn: p0 lower bound} is omitted since
$h_l$ has no intercept by \eqref{eqn: uniadd} of the main paper. Hence, for sufficiently large $n$:
\begin{align}
    \log \Pi\big(B_n(f_0, \tilde\epsilon_n)\big)
    &\ge \log \Pi\Big(B\big(g_0, \tfrac{\tilde\epsilon_n}{2}\big)\Big)
    + \sum_{l=1}^d \log \Pi\Big(B\big(h_{0,l}, \tfrac{\tilde\epsilon_n}{2Ld}\big)\Big)
    \label{eqn: aggregate small ball} \\
    &\ge -\frac{C_2}{\tilde\epsilon_n}\log \frac{\bar N}{\tilde\epsilon_n}
    \ge -C_2(c_0 + 1)\frac{\log n}{\tilde\epsilon_n}
    \ge -C_3 n\tilde\epsilon_n^2,
    \label{eqn: prior mass condition}
\end{align}
where $C_2$ is a constant depending only on $d$, $L$, $g_0(0)$ and
$\zeta_g, \zeta_1, \cdots, \zeta_d$, and $C_3 = C_2(c_0 + 1)$. The first inequality uses
$\|\cdot\|_n \le \|\cdot\|_\infty$ and \eqref{eqn: product of small balls}, the second applies
\eqref{eqn: exact lower bound on small ball probability} to each of the $d+1$ factors with
$N_g, N_l \le \bar N$, and the last two use $\log(\bar N/\tilde\epsilon_n) \le (c_0 + 1)\log n$,
which follows from $\log \bar N \le c_0 \log n$ for some $c_0 > 0$, and
$\log n/\tilde\epsilon_n = \lambda^{-3} n\tilde\epsilon_n^2$.

For the second condition, fix $D > C_3 + 2\sigma_\epsilon^{-2}$ and let:
\begin{equation}
    k_n = \big\lceil A n^{1/3}(\log n)^{-1/3}\big\rceil, \qquad M_n = n^{1/3},
    \qquad \cF_n = \cF_{k_n, M_n},
\end{equation}
with $A > 4\lambda^2 D$. Since $\log(k_n/e\zeta_{\max}) \ge (\log n)/4$ for all large $n$, Lemma
\ref{lem: sieve prior mass} gives:
\begin{align}
    \Pi(\cF_n^c)
    &\le (d + 1)\exp\Big(-k_n \log \frac{k_n}{e\zeta_{\max}}\Big)
    + \big\{(d + 1)\bar N + 2\big\}e^{-n^{2/3}/2} \nonumber \\
    &\le (d + 1)e^{-A n^{1/3}(\log n)^{2/3}/4} + \big\{(d + 1)\bar N + 2\big\}e^{-n^{2/3}/2}
    \le e^{-Dn\tilde\epsilon_n^2},
\end{align}
for sufficiently large $n$, where the last inequality uses:
$A n^{1/3}(\log n)^{2/3}/4 = A n\tilde\epsilon_n^2/(4\lambda^2) > Dn\tilde\epsilon_n^2$,
$\log \bar N \le c_0 \log n$ and $n^{2/3}/(n\tilde\epsilon_n^2) \to \infty$.

For the third condition, Lemma \ref{lem: entropy} with $\eta = \tilde\epsilon_n/8$, whose
hypotheses $\bar N M_n \ge 1$, $k_n \ge 2$ and $\eta \le (d+1)k_n M_n$ hold for sufficiently large $n$,
gives:
\begin{align}
    \log \cC\big(\tilde\epsilon_n/8, \cF_n, \|\cdot\|_n\big)
    &\le (d + 1)k_n \log \bar N
    + \big\{(d + 1)k_n + 1\big\}\log\frac{48d\bar N k_n^2 M_n^2}{\tilde\epsilon_n} \nonumber \\
    &\le \big\{(d + 1)k_n + 1\big\}(2c_0 + 2)\log n
    \le A' n\tilde\epsilon_n^2,
    \label{eqn: entropy condition}
\end{align}
where $A' = 4(d+1)A(c_0 + 1)$. The first inequality uses $\|\cdot\|_n \le \|\cdot\|_\infty$, the
second uses $\log \bar N \le c_0 \log n$ and $48d\bar N k_n^2 M_n^2/\tilde\epsilon_n \le n^{c_0 + 2}$,
and the last uses $(d + 1)k_n + 1 \le 2(d + 1)A n^{1/3}(\log n)^{-1/3}$ and $\lambda \ge 1$. This
verifies \eqref{eqn: three conditions}, and thus completes the proof.
\end{proof}

\section{Details on Posterior Computation}\label{sec: posterior sampler derivation appendix}

In this section, we provide further details of the posterior sampling procedure outlined in Section \ref{sec: computation} of the main paper. Section \ref{ssec: full conditional derivation} derives the full conditional distributions of the Gibbs sampler, Section \ref{ssec: onthefly} describes the on-the-fly normalization step, Section \ref{ssec: complexity} derives the computational complexity of the Gibbs sampler, and Section \ref{ssec: data adaptive hyperparameters} gives details on the specification of prior hyperparameters.

\subsection{Derivation of the full conditional distributions}\label{ssec: full conditional derivation}

In what follows, let $\widetilde y_i = y_i - \delta(\mathbf{x}_i)$ denote the response with the discrepancy term omitted, let $u_i = \sum_{l=1}^d h_l(x_{il})$ denote the link function argument for the $i$-th data point $\mathbf{x}_i$, and let $\bTheta_-$ denote the parameter set $\bTheta$ with the parameter being updated omitted.

We begin with the link function parameters $\bTheta_g$. Conditional on $\bTheta_h$, the link arguments $u_i$ are fixed, so \eqref{eqn: augmented model} of the main paper reduces to a Gaussian linear model in the coefficients, namely $\widetilde y_i = \sum_{j=0}^{N_g}\alpha_j\psi_j^{(g)}(u_i) + \epsilon_i$. Exploiting conjugacy and completing the square in $\alpha_0$ and in $Z_j^{(g)}$ then give \eqref{eqn: fullcondint} and \eqref{eqn: fullcondmag1} of the main paper. For the indicator variables, the two residual sums of squares compared in \eqref{eqn: fullcondind1} of the main paper differ only through the term $Z_j^{(g)}\psi_j^{(g)}(u_i)$, and:
\begin{equation}\label{eqn: indicator odds g}
\sum_{i=1}^n\Big\{\big(r_{ij} - Z_j^{(g)}\psi_j^{(g)}(u_i)\big)^2 - r_{ij}^2\Big\}
= -2Z_j^{(g)}\sum_{i=1}^n \psi_j^{(g)}(u_i)\Big\{r_{ij} - \tfrac{Z_j^{(g)}}{2}\psi_j^{(g)}(u_i)\Big\},
\end{equation}
so that multiplying by $-1/(2\sigma^2_\epsilon)$ recovers the log likelihood ratio appearing there.

We turn next to the additive function parameters $\bTheta_h$. For a fixed choice of $l$ and $j$, let $Z = Z_j^{(l)}$ denote the magnitude, set $s_i = \psi_j^{(l)}(x_{il})$ for $l \in \cI$ (and $s_i = \phi_j^{(l)}(x_{il})$ for $l \notin \cI$), and let $w_i = u_i - \beta_j^{(l)} s_i$ be the contribution of the remaining coefficients, so that $u_i = w_i + Z s_i$ when the coefficient is active. Since $s_i \ge 0$, $u_i$ is non-decreasing in $Z$. We consider the magnitude $Z$ first, since the indicator odds \eqref{eqn: fullcondind2} of the main paper are assembled from the quantities $\{\omega_v\}$ that the magnitude derivation produces.

Note that the fitted value at $\mathbf{x}_i$ can be written as a function of $Z$:
\begin{equation}\label{eqn: fitted in Z}
g(u_i) = \alpha_0 + \sum_{k=1}^{N_g}\alpha_k\,\varrho\big\{N_g(w_i + s_iZ) - k + 1\big\},
\end{equation}
a non-negative combination of clamped ramps in $Z$, so $Z \mapsto g(u_i)$ is piecewise linear. The $k$-th ramp bends where its argument reaches $0$ and where it reaches $1$, so over $k = 1,\cdots,N_g$ the bends of \eqref{eqn: fitted in Z} occur at:
\begin{equation}\label{eqn: breakpoints}
p_{i,m} = \frac{m/N_g - w_i}{s_i}, \qquad m = 0,1,\cdots,N_g, \quad i: s_i > 0,
\end{equation}
i.e., the values of $Z$ at which the $i$-th link argument reaches the $m$-th knot of $g$. Here, design points with $s_i = 0$ do not contribute, since their $u_i$ does not move with $Z$. Collecting the breakpoints \eqref{eqn: breakpoints} over all $i$ and $m$ and sorting them into $t_1 < \cdots < t_M$ divides $\bbR$ into intervals $\cB_0, \cdots, \cB_M$, with $\cB_v = (t_v,t_{v+1}]$, $t_0 = -\infty$ and $t_{M+1} = \infty$; there are at most $n(N_g+1)$ of them, so $M$ is finite.

Within a single $\cB_v$, no link argument crosses a knot, so each $u_i$ stays in one cell of $g$ and every fitted value is linear in $Z$. To identify the two coefficients, note that for $z \ge 0$ with cell index $m = \min\{\lfloor N_g z\rfloor, N_g\}$, clamping gives:
\begin{equation}\label{eqn: psi step}
\psi_k^{(g)}(z) = \mathds{1}\{k \le m\} + (N_gz - m)\,\mathds{1}\{k = m+1\},
\qquad
g(z) = G_m + (N_gz - m)\,\alpha_{m+1},
\end{equation}
where $G_m := \alpha_0 + \sum_{k\le m}\alpha_k$ is the value of $g$ at the knot $m/N_g$ and $\alpha_{N_g+1} = 0$. The same holds for $\psi_k^{(l)}$ and $h_l$ with $(N_g,\alpha,G_m)$ replaced by $(N_l,\beta^{(l)},H^{(l)}_m)$, $H^{(l)}_m = \sum_{k\le m}\beta^{(l)}_k$, so that evaluation requires only the cell index and the cumulative sum, never the individual basis functions. Letting $m_{iv}$ index the cell containing $u_i$ for $Z \in \cB_v$ and applying \eqref{eqn: psi step}, we have:
\begin{equation}\label{eqn: affine pieces}
g(w_i + Zs_i) = A_{iv} + B_{iv}Z,
\qquad
A_{iv} = G_{m_{iv}} + (N_gw_i - m_{iv})\alpha_{m_{iv}+1},
\qquad
B_{iv} = N_g\alpha_{m_{iv}+1}s_i .
\end{equation}

Next, expanding $-\frac{1}{2\sigma^2_\epsilon}\sum_i(\widetilde y_i - A_{iv} - B_{iv}Z)^2$ shows that the conditional log-likelihood of $Z$ yields the quadratic form $-(a_vZ^2 + b_vZ + c_v)$ on $\cB_v$, with:
\begin{equation}\label{eqn: abc coefficients}
a_v = \frac{1}{2\sigma^2_\epsilon}\sum_{i=1}^n B_{iv}^2,
\qquad
b_v = -\frac{1}{\sigma^2_\epsilon}\sum_{i=1}^n (\widetilde y_i - A_{iv})B_{iv},
\qquad
c_v = \frac{1}{2\sigma^2_\epsilon}\sum_{i=1}^n (\widetilde y_i - A_{iv})^2 .
\end{equation}

Consider first the magnitude parameter. Multiplying the likelihood on $\cB_v$ by the prior $\cN^{+}(0, \tau_l^2)$, the posterior is proportional to the exponent of $-(Z-\widehat\mu_v)^2/(2\widehat\sigma_v^2) + \widehat\mu_v^2/(2\widehat\sigma_v^2) - c_v$, where:
\begin{equation}\label{eqn: piecewise gaussian params}
    \widehat{\sigma}_v^2 = \Big(2a_v + \tau_l^{-2}\Big)^{-1}, \quad
    \widehat{\mu}_v = -b_v \widehat{\sigma}_v^2, \quad
    \omega_v = \widehat{\sigma}_v\, e^{\widehat{\mu}_v^2/(2\widehat{\sigma}_v^2) - c_v}
    \Big\{\Phi\Big(\frac{R_v - \widehat{\mu}_v}{\widehat{\sigma}_v}\Big) - \Phi\Big(\frac{L_v - \widehat{\mu}_v}{\widehat{\sigma}_v}\Big)\Big\},
\end{equation}
and $[L_v,R_v] = \cB_v \cap [0,\infty)$. Integrating the resulting kernel over $[L_v,R_v]$ gives:
\begin{equation}\label{eqn: exp quadratic integration over interval}
    e^{\widehat\mu_v^2/(2\widehat\sigma_v^2) - c_v}\int_{L_v}^{R_v}
    e^{-(Z-\widehat\mu_v)^2/(2\widehat\sigma_v^2)}\,dZ = \sqrt{2\pi}\,\omega_v ,
\end{equation}
so the total mass on $\cB_v$ is proportional to $\omega_v$, and normalizing the $\omega_v$ across $v$ yields \eqref{eqn: fullcondmag2} of the main paper, a finite mixture of normals truncated to the intervals $[L_v,R_v]$. For a non-monotone input $l \notin \cI$, the weights have no probability mass on zero, so $\beta^{(l)}_j = Z$ is always active; the derivation above applies verbatim with $I_j^{(l)} \equiv 1$ and gives \eqref{eqn: fullcondmag3} of the main paper.

Consider next the indicator parameter for $l \in \cI$. Write $\mathcal{L}(Z) =\exp\{-(a_vZ^2+b_vZ+c_v)\}$ for $Z \in \cB_v$, which is the likelihood function up to the factor $(2\pi\sigma_\epsilon^2)^{-n/2}$. Under $I_j^{(l)} = 1$, the magnitude is integrated over its $\cN^+(0,\tau_l^2)$ prior, whose density is $\sqrt{2/(\pi\tau_l^2)}\,e^{-Z^2/2\tau_l^2}$. Splitting the integral over the intervals $\cB_v \cap [0,\infty)$ and applying \eqref{eqn: exp quadratic integration over interval} on each interval, we have:
\begin{equation}\label{eqn: marginal likelihood}
\int_0^\infty \mathcal{L}(Z)\sqrt{\tfrac{2}{\pi\tau_l^2}}\,e^{-Z^2/2\tau_l^2}\,dZ
= \sqrt{\tfrac{2}{\pi\tau_l^2}}\,\sqrt{2\pi}\sum_{v=0}^M\omega_v
= \frac{2}{\tau_l}\sum_{v=0}^M\omega_v .
\end{equation}
Under $I_j^{(l)} = 0$, the coefficient vanishes, so the likelihood is evaluated at $Z = 0$ and equals $\mathcal{L}(0) = e^{-c_{v_0}}$, where $v_0$ is the index with $0 \in \cB_{v_0}$. Thus, the log-ratio of \eqref{eqn: marginal likelihood} to $e^{-c_{v_0}}$, added to the log prior odds $\log\{\zeta_l/(N_l-\zeta_l)\}$, is exactly \eqref{eqn: fullcondind2} of the main paper.

Finally, we investigate the remaining parameters $\boldsymbol{\Theta}_\delta$ and $\sigma^2_\epsilon$. Note that $\bdelta \sim \mathcal{N}(\mathbf{0},\sigma^2_\delta\mathbf{K})$ and $\mathbf{r}\mid\bdelta \sim \mathcal{N}(\bdelta,\sigma^2_\epsilon\mathbf{I}_n)$, so Gaussian conditioning gives \eqref{eqn: fullconddelta} of the main paper, where the covariance:
\begin{equation}\label{eqn: delta covariance}
\sigma^2_\delta\mathbf{K} - \sigma^4_\delta\mathbf{K}\mathbf{M}^{-1}\mathbf{K}
= \sigma^2_\delta\mathbf{K}\mathbf{M}^{-1}(\mathbf{M}-\sigma^2_\delta\mathbf{K})
= \sigma^2_\epsilon\sigma^2_\delta\mathbf{K}\mathbf{M}^{-1}
\end{equation}
is symmetric since $\mathbf{K}$ and $\mathbf{M}$ commute. Equations \eqref{eqn: fullcondgpvar} and \eqref{eqn: fullcondnoisevar} of the main paper follow from Inverse-Gamma conjugacy.

\subsection{On-the-fly normalization}
\label{ssec: onthefly}

While our Gibbs sampler (Algorithm \ref{alg: gibbs} of the main paper) works without the following on-the-fly normalization adjustment, we find that such an adjustment can improve parameter identifiability and MCMC mixing. On parameter identifiability, note that a function $f$ in the transformed additive form with link $g$ and additive components $\{h_l\}_{l=1}^d$ (see Equation \eqref{eqn: generalized additive model} of the main paper) can equivalently be modeled by a different link $g'$ and additive components $\{h_l'\}_{l=1}^d$, where $h_l' = \gamma_0 + \gamma h_l$ with $\gamma > 0$ and $g'$ is adjusted accordingly. This non-identifiability is not a big issue if the goal is to infer the overall function $f$, but improving parameter identifiability can improve MCMC mixing performance \cite{wang2026mcmc}.

One way to improve identifiability is to add a constraint on basis weights $\{\beta_j^{(l)}\}$, e.g.:
\begin{equation}\label{eqn: beta star sum 1 condition}
    \sum_{l \in \cI} \sum_{j=1}^{N_l} \beta_j^{(l)} + \sum_{l \notin \cI} \underset{j = 1, \cdots, N_l}{\max} \beta_j^{(l)} = 1.
\end{equation}
Here, the use of summed basis weights in the first term vs. the use of maximum basis weights in the second term is due to the nature of basis functions for $l \in \mathcal{I}$ vs. $l \notin \mathcal{I}$. This constraint, which we adopt in our implementation, has the added benefit of ensuring $\sum_{l=1}^d h_l$ remains within $[0,1]$, which is desirable for modeling the monotone link $g$ via \eqref{eqn: monotone link function} of the main paper.

For posterior sampling, we can enforce the constraint \eqref{eqn: beta star sum 1 condition} via a simple ``on-the-fly'' adjustment in our Gibbs sampler. In particular, one first samples the basis indicators and magnitudes for the weights $\{\beta_{j}^{(l)}\}$ using their full conditional distributions, then normalizes these sampled weights by dividing them by:
\begin{equation}\label{eqn: S def} 
S = \sum_{l \in \cI} \sum_{j=1}^{N_l} \beta_j^{(l)} + \sum_{l \notin \cI} \underset{j = 1, \cdots, N_l}{\max} \beta_j^{(l)},
\end{equation}
which ensures the constraint \eqref{eqn: beta star sum 1 condition} holds. Similar on-the-fly adjustments have been used in spatial statistics \cite{rue2005gaussian,mak2016regional} to enforce spatial effect constraints for parameter identifiability. While the distribution sampled with such an on-the-fly adjustment can differ from the posterior distribution under the constraint, this difference is typically minimal in practice (see \cite{paciorek2009technical}).

With the sampled basis weights $\{\beta_j^{(l)}\}$ normalized, it is also desirable to ensure the sampled function $f$ remains unchanged within the current MCMC iteration. To achieve this, after normalizing $\{\beta_j^{(l)}\}$, we stretch the link function $g$ as $g(u) \leftarrow g(Su)$, where $S$ is the normalization factor in \eqref{eqn: S def}. Since $g$ is piecewise linear with $g(m/N_g) = \alpha_0 + \sum_{j \le m} \alpha_j$ satisfying $\alpha_m = g(m/N_g) - g((m-1)/N_g)$, we can adjust the link basis weights $\alpha_m$ as follows: 
\begin{equation}\label{eqn: alpha rescale} 
\alpha_m \leftarrow g\Big( \frac{S m}{N_g} \Big) - g\Big( \frac{S (m-1)}{N_g} \Big), \qquad m = 1, \cdots, N_g.
\end{equation}
This keeps the sampled function $f$ the same at the link function knots $\{S m /N_g\}_{m=1}^{N_g}$ after the on-the-fly normalization of basis weights $\{\beta_j^{(l)}\}$.

\subsection{Computational complexity}\label{ssec: complexity}

We now derive the computational complexity for each iteration of the Gibbs sampler outlined in Algorithm \ref{alg: gibbs} of the main paper.

At the start of each Gibbs iteration, we first compute the cumulative sums $\{G_m\}$ and $\{H^{(l)}_m\}$ in \eqref{eqn: psi step}, which incurs $\cO(N_g + \sum_lN_l)$ work. This enables a single evaluation of $g$ or $h_l$ with $\cO(1)$ work, and a change in one coefficient to propagate to $\mathbf{u}$ and the fitted values with $\cO(n)$ work. %

The Gibbs updates for parameters in the set $\bTheta_g$ require $\cO(nN_g)$ work. The intercept term requires the computation of $\sum_ir_{i0}$ and a residual update, and each $(I_j^{(g)},Z_j^{(g)})$ requires the computation of $\sum_i\psi_j^{(g)}(u_i)r_{ij}$ and $\sum_i\psi_j^{(g)}(u_i)^2$, all of which incur $\cO(n)$ work.

The Gibbs updates for parameters in the set $\bTheta_h$ require $\cO\{n\sum_lN_l + M_{\rm tot}\log n\}$ work, where $M_{lj}$ is the number of breakpoints \eqref{eqn: breakpoints} for given $l$ and $j$, and $M_{\rm tot} = \sum_l\sum_jM_{lj}$. The first term $\cO(n \sum_l N_l)$ is for computing $\{s_i\}$ and $\{w_i\}$, and for updating $\{u_i\}$ for each $1 \le i \le n$. The second term is for sorting the breakpoints $\{p_{i, m}\}$ over all $i$ and $m$, where we know $\{p_{i, m}\} $ over $m$ are sorted for each $i$: this involves merging $n$ sorted lists with length $N_g +1$, which costs $\cO(M_{lj}\log n)$. Summing this over all $l$ and $j$ gives $\cO\{M_{\rm tot}\log n\}$ work. Note that the quadratic coefficients \eqref{eqn: abc coefficients} are updated rather than recomputed, since consecutive intervals differ by a single term, which makes each step $\cO(1)$ instead of $\cO(n)$. 

Note that the total number of breakpoints $M_{\rm tot}$ is $ \cO\{nN_g\sum_lN_l\}$, as $M_{lj} \le n(N_g+1)$ since each $Z_j^{(l)}$ includes at most $n(N_g + 1)$ breakpoints from \eqref{eqn: fitted in Z}. The bound is tight for $l\in\cI$, where $\psi_j^{(l)}$ is supported on $[(j-1)/N_l,1]$ and $\sum_jM_{lj}\asymp nN_gN_l/2$. For $l\notin\cI$, the hat basis has support of width $2/(N_l-1)$, so $s_i\neq0$ for only $\cO(n/N_l)$ indices and $\sum_jM_{lj} = \cO(nN_g)$ irrespective of $N_l$. Therefore, the upper bound $\cO\{nN_g\sum_lN_l\}$ on $M_{\rm tot}$ becomes tight as more inputs are monotone. This, combined with the above analysis, yields a total work of $\cO\{ nN_g(\sum_lN_l)\log n\}$.

Finally, the Gibbs updates for parameters in the set $\bTheta_\delta$, together with $\sigma^2_\epsilon$, require $\cO(n^3)$ work. The noise variance requires only a residual sum of squares calculation at $\cO(n)$ work. Caching the eigendecomposition $\mathbf{K}=\mathbf{U}\Lambda\mathbf{U}^\top$ reduces the draws of $\bdelta$ and $\sigma^2_\delta$ to $\cO(n^2)$ work, since $\mathbf{M}^{-1} = \mathbf{U}(\sigma^2_\delta\Lambda+\sigma^2_\epsilon\mathbf{I})^{-1}\mathbf{U}^\top$ and $\bdelta^\top\mathbf{K}^{-1}\bdelta = \|\Lambda^{-1/2}\mathbf{U}^\top\bdelta\|^2$ reuse the same factors as $\sigma^2_\delta$ varies. The update for $\xi$ requires $\cO(n^3)$ work, since a proposal $\xi'$ changes $\mathbf{K}$ itself and requires a fresh factorization of the covariance matrix.

Combining the work for the three parts above, the computational complexity claimed in Section \ref{ssec: alg} of the main paper then follows.

\subsection{Specification of prior hyperparameters}\label{ssec: data adaptive hyperparameters}

We provide further details here on the specification of the prior hyperparameters $\mu_0$, $\tau^2_g$, $\tau^2_0$, $\tau^2_1, \cdots, \tau^2_d$, which control the location and scale of the link function $g$ and the additive components $h_1, \cdots, h_d$.

Consider first $\mu_0$, $\tau^2_g$ and $\tau^2_0$, which are hyperparameters for the link function $g$. In practice, given limited data for our surrogate modeling applications, we find that a data-dependent prior on such hyperparameters can be effective. Recall that $\mu_0$ and $\tau^2_0$ are the prior mean and variance for the link intercept $\alpha_0$, which models the minimum of the black-box function $f$ ignoring the GP discrepancy term $\delta$. In our implementation, we set $\mu_0= \min_i y_i$ and $\tau_0^2 = \mathrm{Var}(\mathbf{y})/n$, where $\mathbf{y} = (y_1, \cdots, y_n)$. This focuses the prior of $\alpha_0$ on the minimum of the data, with such a prior concentrating as $n$ grows. Recall that $\tau^2_g$ reflects the scale of the link function $g$. In our implementation, we set $\tau^2_g$ to satisfy $E \{ (g(1) - g(0))^2\} = R_y^2$ where $R_y = \max \{\text{range}(\mathbf{y}), \text{sd}(\mathbf{y}), 1\}$, to capture the scale of the observed data.

Consider next $\tau^2_1, \cdots, \tau^2_d$, which are hyperparameters that reflect the scale of the additive components $h_1, \cdots, h_d$. Recall from Section \ref{ssec: alg} of the main paper that these hyperparameters should be set such that $\sum_{l=1}^d h_l$ is concentrated on $[0,1]$ with high probability. Similar to \cite{chipman2010bart} (where a similar soft constraint is needed), we set $\tau_l^2$ to satisfy the second moment condition $\mathbb{E}[\max_x h_l^2(x)] = 1/d$. For $l \in \cI$, since:
\begin{equation}\label{eqn: second moment rule}
\bbE\Big[\Big(\sum_{j=1}^{N_l}\beta_j^{(l)}\Big)^2\Big]
= N_l\cdot\frac{\zeta_l}{N_l}\tau_l^2 + N_l(N_l-1)\frac{\zeta_l^2}{N_l^2}\cdot\frac{2\tau_l^2}{\pi}
= \tau_l^2\Big\{\zeta_l + \frac{2(N_l-1)\zeta_l^2}{\pi N_l}\Big\}, 
\end{equation}
this means $\tau_l^2$ can be set as:
\begin{equation}\label{eqn: tau additive monotone}
\tau_l^2 = \Big[ d \Big\{ \zeta_l + \tfrac{2(N_l-1)\zeta_l^2}{\pi N_l} \Big\} \Big]^{-1} \quad \text{for $l \in \mathcal{I}$}.
\end{equation}
For $l \notin \cI$, $h_l(x)$ is a convex combination of at most two adjacent coefficients. Thus, we use a simple coefficient-wise calibration $\bbE\{(\beta_j^{(l)})^2\}=1/d$, which gives $\tau_l^2=1/d$. With this specification, we find that the fitted $\sum_{l=1}^d h_l$ concentrates on $[0,1]$ in numerical experiments, as desired.

\end{document}